\documentclass[conference,letterpaper]{IEEEtran}
\IEEEoverridecommandlockouts

\input{aux/packages}

\newcommand{\bfa}{{\boldsymbol a}}
\newcommand{\bfb}{{\boldsymbol b}}

\newcommand{\bfd}{{\boldsymbol d}}
\newcommand{\bfe}{{\boldsymbol e}}

\newcommand{\bfp}{{\boldsymbol p}}
\newcommand{\bfq}{{\boldsymbol q}}
\newcommand{\bfr}{{\boldsymbol r}}

\newcommand{\bfu}{{\boldsymbol u}}
\newcommand{\bfv}{{\boldsymbol v}}
\newcommand{\bfw}{{\boldsymbol w}}

\newcommand{\bfA}{{\mathbf A}}

\newcommand{\bfD}{{\mathbf D}}

\newcommand{\bfP}{{\mathbf P}}

\newcommand{\cB}{\mathcal{B}}

\newcommand{\cD}{\mathcal{D}}

\newcommand{\cQ}{\mathcal{Q}}
\newcommand{\cR}{\mathcal{R}}
\newcommand{\cS}{\mathcal{S}}

\newcommand{\cU}{\mathcal{U}}

\DeclareMathOperator{\baseToBits}{\mathsf{BaseToBits}}       
\DeclareMathOperator{\baseToUnit}{\mathsf{BaseToUnit}}
\DeclareMathOperator{\binaryEntropy}{h_2}                    
\DeclareMathOperator{\Bin}{Bin}                          
\DeclareMathOperator{\bitsToBase}{\mathsf{BitsToBase}}       
\DeclareMathOperator{\Bern}{Bern}
\DeclareMathOperator{\Copy}{\mathsf{Copy}}
\DeclareMathOperator{\dH}{d_H}                               
\DeclareMathOperator{\entropy}{H}                            
\DeclareMathOperator{\mutualInfo}{I}                         
\newcommand{\KL}[2]{D_{\mathrm{KL}}\left(#1 \,\middle\|\, #2\right)} 
\DeclareMathOperator{\PCRandSequence}{\mathsf{PCRandSequence}}
\DeclareMathOperator{\Reconstruct}{\mathsf{Reconstruct}}
\DeclareMathOperator{\unitToBase}{\mathsf{UnitToBase}}

\newcommand{\algoName}[1]{\ensuremath{\mathsf{#1}}}
\newcommand{\bigo}[1]{\mathcal{O}\left(#1\right)}                                   
\newcommand{\define}{\triangleq}

\newcommand{\getrand}{\xleftarrow{\smash{\raisebox{-0.4ex}{\tiny$\mathsf{\$}$}}}}   
\newcommand{\indicator}[1]{\mathbbm{1}\left({#1}\right)}                            
\newcommand{\nth}{\ifmmode ^\mathrm{th}    \else \textsuperscript{th}\xspace    \fi}
\newcommand{\set}[1]{\left[{#1}\right]}                                             
\newcommand{\size}[1]{\left \lvert {#1} \right \rvert}                              
\newcommand{\transpose}{\top}

\newcommand{\ans}{\bfA}                                         
\newcommand{\barcode}[1]{\ensuremath{\bfb}_{#1}}                 
\newcommand{\barcodeOfInterest}{\barcode{\f}}                   
\newcommand{\barcodePool}{\cB}

\newcommand{\baseSet}{\cS}           
\newcommand{\choiceRV}{\mathsf{C}}

\newcommand{\dbs}{\bfD}                                          
\newcommand{\dbsPool}{\cD}                                          
\newcommand{\dbsRV}{\mathsf{D}}
\newcommand{\hbarcode}[1]{\ensuremath{\widehat{\bfb}}_{#1}}      
\newcommand{\hrec}[1]{\ensuremath{\widehat{\bfd}}_{#1}}          
\newcommand{\f}{f}                                             
\newcommand{\indexRV}{\mathsf{F}}                                        
\renewcommand{\k}{k}                                           
\newcommand{\kterm}{c_{\k}}
\renewcommand{\l}{\ensuremath{\ell}}                            
\newcommand{\leakage}{\ensuremath{\varepsilon}}
\newcommand{\lbc}{\ensuremath{\ell}_{\mathrm{bc}}}              
\newcommand{\lpl}{\ensuremath{\ell}_{\mathrm{pl}}}              
\newcommand{\lu}{\ensuremath{\ell}_{\mathrm{u}}}              
\newcommand{\m}{m}                                              
\newcommand{\n}{n}                                              
\newcommand{\p}{p}
\newcommand{\payload}[1]{\ensuremath{\bfp}_{#1}}                 
\newcommand{\pf}{\rho}                                          
\newcommand{\PIRwithXOR}{\algoName{PIRwithXOR}}
\newcommand{\q}{q}
\newcommand{\queriedRecords}{\dbs'}
\newcommand{\query}{\bfq}                                       
\newcommand{\queryPool}{\cQ}
\newcommand{\queryRV}{\mathsf{Q}}                                        
\renewcommand{\r}{r}
\newcommand{\randomPool}{\cR}                                   
\newcommand{\randomPoolCopy}{\randomPool'}                      
\newcommand{\randomPoolRV}{\mathsf{R}}
\newcommand{\randomPoolStar}{\randomPool_\f}                     
\newcommand{\rec}[1]{\ensuremath{\bfd}_{#1}}                     

\newcommand{\un}{u}
\newcommand{\universalSet}{\cU}
\newcommand{\universalPrimer}{\bfw}

\definecolor{bleudefrance}{rgb}{0.19, 0.55, 0.91}

\makeatletter
\newcommand{\colfigwidth}[2]{\if@twocolumn #2\else #1\fi\linewidth}
\makeatother
\providecommand{\parstart}[2]{#1#2}

\begin{document}
\title{Private Information Retrieval for Large-Scale DNA-Based Data Storage}

\author{%
  \IEEEauthorblockN{Gökberk Erdo\u{g}an\IEEEauthorrefmark{1}\IEEEauthorrefmark{2},
                    Daniella Bar-Lev\IEEEauthorrefmark{1}\IEEEauthorrefmark{3},
                    Rawad Bitar\IEEEauthorrefmark{2},
                    Antonia Wachter-Zeh\IEEEauthorrefmark{2},
                    Zohar Yakhini\IEEEauthorrefmark{4}\IEEEauthorrefmark{5}}
\thanks{\IEEEauthorrefmark{1}equal contribution;
    \IEEEauthorrefmark{2}School of Computation, Information and Technology, Technical University of Munich;
  \IEEEauthorrefmark{3}Department of Mathematics, Universität Zürich; \IEEEauthorrefmark{4}Faculty of Computer Science,
                    Technion; and \IEEEauthorrefmark{5}School of Computer Science,
                    Reichman University.}
  \thanks{Emails: \{gokberk.erdogan, rawad.bitar, antonia.wachter-zeh\}@tum.de, daniella.bar-lev@math.uzh.ch, zohary@cs.technion.ac.il}
  \thanks{This work was initiated during the Dagstuhl Seminar 24511 on ``Coding Theory and Algorithms for Emerging Technologies in Synthetic Biology". This work was funded by the European Union (DiDAX, 101115134). Views and opinions expressed are however those of the author(s) only and do not necessarily reflect those of the European Union or the European Research Council Executive Agency. Neither the European Union nor the granting authority can be held responsible for them. The work of D. Bar-Lev was also supported by Schmidt Sciences and by the Swiss National Science Foundation under grant number 212865.}
  \thanks{Fig.~\ref{fig:dataset_setup}--Fig.~\ref{fig:xor_response} were created in BioRender. Bar-Lev, D. (2026).}
}

\maketitle


\begin{abstract}
    We investigate Private Information Retrieval (PIR) in the context of synthetic DNA-based data storage. While PIR is a well-studied primitive for digital databases, extending it to DNA-based databases presents unique challenges arising from biochemical query mechanisms and their complexity. We propose two approaches for adapting two-server PIR protocols to DNA-based storage, balancing privacy, efficiency, and feasibility. These approaches illustrate how information-theoretic privacy trade-offs manifest in DNA-based storage systems.

\end{abstract}

\section{Introduction and Motivation}
\parstart{P}{rivate} Information Retrieval (PIR) allows a user to retrieve a specific record from a database without revealing which record is accessed~\cite{chor_private_1995}. In the classical setting, PIR protocols assume digital storage with efficient random access and in-memory algebraic computation. Multi-server PIR schemes achieve information-theoretic privacy by distributing queries across non-colluding servers and combining their responses using simple algebraic operations such as XOR. These protocols have been extensively studied and optimized for electronic storage systems, e.g.,~\cite{sun_capacity_2017, banawan_capacity_2018, sun_robust_2018, ulukus_private_2022, hollanti_2017_private, kadhe_2017_private, tajeddine_2016_private, raviv_2018_private}.

In parallel, synthetic DNA-based data storage has emerged as a promising medium for ultra-dense and long-term archival storage, see e.g.~\cite{church2012next, goldman2013towards, grass2015robust, bornholt2016dna, yazdi2017portable, erlich2017dna, organick2018random, bar2025scalable, anavy2019data, preuss2024efficient}. In such systems, digital data is encoded into large pools of synthetic DNA molecules, with millions or billions of distinct strands collectively representing the database. DNA-based storage offers storage densities and durability far beyond conventional electronic media~\cite{erlich2017dna, grass2015robust, anavy2019data}, but its access model differs fundamentally from that of digital memory~\cite{tabatabaei2015rewritable, organick2018random}. This mismatch motivates a careful re-examination of information retrieval in the context of DNA-based storage. In particular, we focus on PIR from DNA-based data, which, to the best of our knowledge, has only been studied in the independent and concurrent work of Wang et al.~\cite{wang2026sequential} for the single-server setting. 

We consider a DNA-based storage system with a database replicated over multiple servers, and a user privately retrieving a certain record. Each server stores the database as a pool of synthetic DNA strands and can answer queries by performing biochemical operations, followed by sequencing and digital post-processing. The user may also perform biochemical operations, most notably, simple DNA synthesis. A key parameter of the model is the user's biochemical capability, which varies across the scenarios we consider and determines the privacy guarantees and the achievable efficiency, as we shall explain.

In classical two-server PIR~\cite{chor_private_1995}, a database $\dbs$ containing $\n$ files is stored across the servers. To retrieve file $\bfd_\f$, the user draws a random sequence ${\bfu \in \mathbb{F}_2^\n}$ and sends the query $\bfq_1 = \bfu$ to server $1$ and ${\bfq_2 = \bfu \oplus \bfe_\f}$ to server $2$, where $\bfe_\f$ is the $\f\nth$ standard basis vector and $\oplus$ is addition modulo 2, i.e., the XOR operation. The servers respectively respond with the dot product $\bfa_i = \langle \dbs, \bfq_i \rangle$, and the user computes $\bfd_\f = \bfa_1 \oplus \bfa_2$.

A direct implementation of classical PIR in DNA storage has a prohibitive cost. Each server would need to synthesize a \emph{primer} {(short DNA strand)} for every selected record per query, that is $\mathcal{O}(\n)$ individual synthesis operations. Moreover, unlike electronic storage, where servers can compute algebraic combinations of records in memory, DNA-based storage does not support in-pool computation: any algebraic processing, such as XOR or linear combinations of records, can only be performed after sequencing. PIR schemes must therefore be redesigned so that privacy does not rely on pre-sequencing algebraic manipulation. At the same time, synthesizing primers with \emph{random} sequences is inexpensive~\cite{anavy2019data, luescher2024chemical}, and \emph{polymerase chain reaction} (PCR) amplification time is, in principle, independent of the pool size~\cite{organick2018random}. These properties are the key features of DNA-based storage that our protocols exploit: queries are constructed primarily from randomly synthesized primers, with targeted synthesis limited to a single sequence.

%
%
We investigate how multi-server PIR protocols can be adapted to DNA-based storage under these constraints. We focus on a single core approach: exploiting inexpensive random primer pool synthesis. We introduce two novel PIR protocols for DNA-based storage and consider three scenarios of increasing generality that differ in the biochemical capabilities assumed for the user, the privacy guarantees achieved, and the number of servers used. These protocols and scenarios illustrate the fundamental trade-offs involved in PIR from large-scale DNA-based storage systems.
We assume a noiseless setting unless stated otherwise, deferring the important analysis of DNA-specific error sources such as strand dropouts, amplification bias, and sequencing noise to future work.


\section{Preliminaries}
\subsection{Notation}
Sequences are denoted by bold lowercase letters and matrices by bold uppercase letters. For a sequence $\bfv$, $\bfv[i]$ denotes its $i\nth$ element. The Hamming distance between $\bfu$ and $\bfv$ of the same length is $\dH(\bfu, \bfv) = |\{i : \bfu[i] \neq \bfv[i]\}|$. Sets are denoted by uppercase calligraphic letters. For a positive integer $k$, $\set{k}$ denotes the set $\{1, \dots, k\}$. We consider the DNA alphabet $\baseSet = \{A, C, G, T\}$. A length-$\l$ \emph{DNA sequence} is an element of $\baseSet^{1 \times \l}$. We denote the Bernoulli distribution with success probability $\p$ by $\Bern(\p)$ and 
the binomial distribution with number of trials $\nu$ and success probability $\p$ as $\Bin(\nu, \p)$. 
The operator $\getrand$ denotes sampling from a set uniformly at random, or sampling from a probability distribution. All logarithms are base $2$. The entropy of a random variable $\mathsf{X}$ is denoted by $\entropy(\mathsf{X})$ and $\binaryEntropy(p)$ denotes the binary entropy. The mutual information between two random variables $\mathsf{X}$ and $\mathsf{Y}$ is denoted by $\mutualInfo(\mathsf{X};\mathsf{Y})$. The Kullback--Leibler (KL) divergence between two Bernoulli distributions $P_1 = \Bern(p)$ and $P_2 = \Bern(q)$ with $p, q \in [0, 1]$ is 
\[
\KL{\Bern(p)}{\Bern(q)} = p \log \frac{p}{q} + (1 - p) \log \frac{1 - p}{1 - q}.
\]

\textbf{Copying pools of sequences.}
Given a pool of sequences $\randomPool$, $\Copy(\cdot)$ generates a copy of the pool $\randomPool$ as $\randomPoolCopy = \Copy(\randomPool)$. Depending on the setting, the copy is either \textit{perfect} or imperfect\footnote{Here, we assume that $\Copy(\cdot)$ is a stochastic set operation as a proxy of the noisy chemical process. However, we do not address the specific effect of PCR bias on concentrations in the copy operation.}. The copy is \textit{imperfect}, if $ \forall \bfr \in \randomPool$, the following hold: 
\[
\Pr[\bfr \in \randomPoolCopy \mid \bfr \in \randomPool] = 1-\delta \quad \text{and} \quad \Pr[\bfr \in \randomPoolCopy \mid \bfr \notin \randomPool] = 0,
\]
where $0<\delta\leq1$. The copy is \emph{perfect} if $\delta=0$. 

\textbf{Base-to-unit encoding.} We define a map $\baseToUnit(\cdot): \baseSet \mapsto \{0,1\}^{4 \times 1}$ as $A \rightarrow 1000$, $C \rightarrow 0100$, $G \rightarrow 0010$, $T \rightarrow 0001$.
The reverse mapping $\unitToBase(\cdot): \{0,1\}^{4 \times 1} \rightarrow \baseSet$ maps the length-4 unit vectors to the corresponding bases.

\textbf{Base-to-bit encoding.} We define a map $\baseToBits(\cdot): \baseSet \mapsto \{0,1\}^{2 \times 1}$ as $A \rightarrow 00$, $C \rightarrow 01$, $G \rightarrow 10$, $T \rightarrow 11$.
The reverse mapping $\bitsToBase(\cdot): \{0,1\}^{2 \times 1} \rightarrow \baseSet$ maps the 2-bit binary sequences back to their corresponding bases. 

The maps $\baseToUnit(\cdot)$ and $\baseToBits(\cdot)$ work base-by-base on sequences of bases.

\subsection{Private Information Retrieval}
Private Information Retrieval (PIR)~\cite{chor_private_1995} is a primitive that allows a user to retrieve a record $\rec{\f}$ from a database $\dbs = (\rec{1}^\transpose, \dots, \rec{\n}^\transpose)^\transpose$ stored at $\m$ servers, while keeping the index $\f$ private. 
The user generates a query $\query_j$ for each server $ j$ and sends it to that server. Each server processes its query on $\dbs$ and returns an answer $\bfa_j$. The queries are constructed such that any set of $z<m$ servers learns nothing about $\f$ and that the user can reconstruct $\rec{\f}$ from the collection of answers.


\textbf{Correctness.} Let $\dbsRV_{\f}$ be the random variable for record $\rec{\f}$. A PIR protocol is said to be \emph{correct} if
\[
\entropy(\dbsRV_{\f} \mid \bfa_1, \bfa_2, \dots, \bfa_m, \query_1, \query_2, \dots, \query_m) = 0. 
\]

\textbf{Privacy.} Let $\indexRV$ be the random variable for the queried file index $\f$, $\queryRV_i$ be the random variable for a query $\query_i$ and $z<m$ be a privacy parameter. A PIR protocol achieves \emph{perfect information-theoretic (IT)} privacy if $\mutualInfo(\queryRV_{i_1},\dots,\queryRV_{i_z};\indexRV) = 0$ for any $z$ queries. 
A PIR protocol achieves \emph{weak IT} privacy if $\mutualInfo(\queryRV_{i_1},\dots,\queryRV_{i_z};\indexRV) \leq \leakage$ for a desired $\leakage>0$. The term $\leakage$ is defined as the \textit{information leakage}, also analyzed in \cite{lin_weakly-private_2019, qian_improved_2022, bitar2024sparsity}.

\textbf{A basic two-server IT-PIR scheme.} 
Consider the two-server PIR scheme mentioned in the introduction, which forms the basis of our constructions in Section~\ref{sec:adaptingbasicpir}. Here, $m=2$ and $z=1$. IT-privacy follows from the fact that both $\bfq_1$ and $\bfq_2$ are uniformly distributed over $\{0,1\}^{\n \times 1}$ independently of $\f$. Thus, neither query reveals any information about $\f$.



\subsection{DNA-Based PIR Model}
\label{sec:dna-model}
We consider a multi-server PIR with replicated data for synthetic DNA-based data storage. Under this setting, each server stores the database $\dbs$ as a pool of synthetic DNA strands and answers queries using biochemical operations followed by sequencing and digital post-processing. A key parameter of the model is the user's biochemical capability, which determines how queries are generated. We progressively introduce the different scenarios with varying user biochemical capabilities in Section~\ref{sec:adaptingbasicpir}, after presenting our protocols.

\textbf{Strand and database structure.}
In DNA-based storage, digital data is encoded into synthetic DNA molecules called \emph{strands}. A strand is a short sequence over $\baseSet$, typically hundreds to thousands of symbols
long. Large collections of such strands are synthesized and stored together in solution, forming a \emph{pool}. Unlike electronic memory, strands in a pool are unordered and cannot be addressed by position; instead, each strand carries its own address as part of its sequence. Formally, each record $\rec{i}$ is stored as a strand of fixed length $\l = \lbc + \lpl + \lu$ nucleotides, partitioned into three parts:
\[
   \rec{i} = [\barcode{i},\bfp_i,\universalPrimer], \qquad
    \barcode{i} \in \baseSet^{\lbc}, \quad \bfp_i \in \baseSet^{\lpl}, \quad \universalPrimer \in \baseSet^{\lu},
\]
where $\barcode{i}$ is the \emph{barcode} (address), $\bfp_i$ is the \emph{payload} (content), and $\universalPrimer$ is a short universal sequence (serving as a template for a universal reverse primer) common to all strands in the pool. The database is then the pool $\dbsPool = \{\rec{1}, \rec{2}, \dots, \rec{\n}\}.$
Barcodes serve as unique identifiers: $\barcode{i} \neq \barcode{j}$ for all $i \neq j$. We write $\cB = \{\barcode{1}, \dots, \barcode{\n}\}$ for the set of barcodes in $\dbsPool$. We write $\universalSet = \baseSet^{\lbc}$ for the universe of all sequences of length $\lbc$, so $\un\triangleq\size{\universalSet} = 4^{\lbc}$ and $\cB \subseteq \universalSet$. See Fig.~\ref{fig:dataset_setup} for a schematic representation of the database setup.
\begin{figure}[ht]
    \centering
    \includegraphics[width=\colfigwidth{0.45}{0.7}]{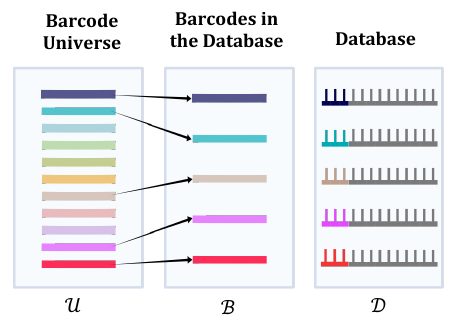}
    \caption{A schematic representation of the database setup. From left to right: the universe $\universalSet$, the set of barcodes in the database $\cB\subseteq\universalSet$, and the corresponding database $\dbsPool$, where the barcodes act as indices of the stored strands. The universal primer $\universalPrimer$ is not depicted as it is the same for all records.}
    \label{fig:dataset_setup}
\end{figure}

\textbf{Selective amplification via PCR.}
Random access to individual strands within a DNA pool is achieved via PCR. In our model, each strand carries a unique barcode at one end and a short universal sequence $\universalPrimer$ at the other. To retrieve the strand with barcode $\barcode{i}$, a primer matching $\barcode{i}$ is introduced into the pool alongside the universal reverse primer. Enzymatic amplification then selectively replicates only the target strand. Importantly, PCR can amplify multiple distinct targets simultaneously by introducing several primers at once, and the amplification time does not scale with the pool size $\n$. This 
is a key feature of DNA-based storage that our protocols exploit.

\textbf{Barcode separation requirement.}
PCR primers bind to their target via hybridization, which tolerates a small number of sequence mismatches. To prevent a primer designed for barcode $\barcode{i}$ from inadvertently co-amplifying a strand with a similar barcode $\barcode{j}$, the barcodes must be sufficiently separated in Hamming distance\footnote{The tolerance of PCR to primer mismatches is more accurately captured by edit (Levenshtein) distance, which accounts for insertions and deletions in addition to substitutions. We use Hamming distance here as a mathematically convenient approximation.}. We refer to such inadvertent co-amplification as PCR primer cross-talk. We require that for all $i \neq j$,  $\dH(\barcode{i}, \barcode{j}) \geq 2\pf + 1$,
where $\pf \geq 1$ is the \emph{proximity tolerance}, that is, the maximum number of mismatches under which PCR may still inadvertently amplify a non-target strand. In coding-theoretic terms, the barcode set $\cB$ must be a code over $\baseSet^{\lbc}$ with minimum distance $2\pf + 1$.
To prevent primer cross-talk, the Hamming spheres of radius $\pf$ around each barcode must be disjoint. By the Hamming (sphere-packing) bound, this restricts the maximum number of usable barcodes to:
\[
    \n =|\cB| \leq \frac{4^{\lbc}}{\displaystyle\sum_{k=0}^{\pf} \binom{\lbc}{k}
    3^k}.
\]
Setting $\pf = 2$, which is typical in practice, as PCR tolerates up to one or two substitutions under standard experimental conditions~\cite{ye2012primer}, we have $\sum_{k=0}^{2}\binom{\lbc}{k} \cdot 3^k > \lbc^2$ for all $\lbc > 0$, yielding
the convenient bound
\begin{equation*}
    \n < \frac{4^{\lbc}}{\lbc^2}.
\end{equation*}

\textbf{Query model and constraints.}
The user wishes to retrieve record $\rec{\f}$ while preserving privacy. To do so, the user generates primers from $\universalSet$ and sends them to the servers. Each server adds the received primers to its pool, performs PCR amplification, sequences the result, and returns a digital response. Any computation on the data is performed only after sequencing. Servers thus operate in a \emph{first-molecular-then-digital} pipeline. 
Concretely, a server can apply PCR to selectively amplify target strands, sequence the resulting material, and compute arbitrarily on the sequenced output.

A direct implementation of the two-server PIR scheme~\cite{chor_private_1995} would require each server to receive a binary query $\bfq \in \{0,1\}^{\n \times 1}$ and synthesize a targeted PCR primer for every index $i$ with $\bfq[i] = 1$. This is infeasible in the DNA setting: on average, half the entries of $\bfq$ are $1$, meaning $\mathcal{O}(\n)$ specific primers must be synthesized per query, which is prohibitively expensive for large databases. At the same time, DNA-based storage offers properties that can be exploited for PIR. Synthesizing \emph{random} primers is cheap, and PCR amplification time is independent of the pool size --- so accessing a record costs the same regardless of how large the database is. Our protocols are designed around these two properties.

For our protocols, we assume a noiseless biochemical setting. Specifically, we assume that PCR amplification is perfect, meaning there is no inadvertent amplification of non-target strands, provided that the barcodes are designed with the appropriate Hamming distance. Furthermore, we assume that $\Copy$ operations are perfect. We briefly discuss the effects of imperfect copying in Scenario~3 and will provide a comprehensive treatment of experimental noise in future work.

\section{DNA-PIR Protocols} 
\label{sec:adaptingbasicpir}
We present two PIR protocols, each utilizing two servers.
The protocols share the same subroutines $\algoName{Setup}$ (Algorithm~\ref{alg:setup}) and $\algoName{Query}$ (Algorithm~\ref{alg:query}). 
The former takes as input the database $\dbsPool$ and outputs the list of barcodes $\barcodePool$ used in the database.
The latter generates queries via a set of biochemical operations. First, the user biochemically generates a random subset $\randomPool \subseteq \universalSet \setminus \{\barcodeOfInterest\}$ of barcode sequences via random (degenerate) synthesis.
Next, a biochemical copy $\randomPoolCopy$ of $\randomPool$ is generated via amplification and physical aliquoting (splitting the liquid volume). The user then biochemically synthesizes the target primer $\barcodeOfInterest$ via standard targeted synthesis, and physically adds it to $\randomPoolCopy$ to obtain $\randomPoolStar =\randomPoolCopy\cup\{\barcodeOfInterest\}$. The pools $\randomPool, \randomPoolStar$ are sent to the two servers in random order. The query generation is depicted in Fig.~\ref{fig:query_prep} and detailed in Algorithm~\ref{alg:query}.
\begin{figure}[ht]
    \centering
    \includegraphics[width=\colfigwidth{0.5}{0.79}]{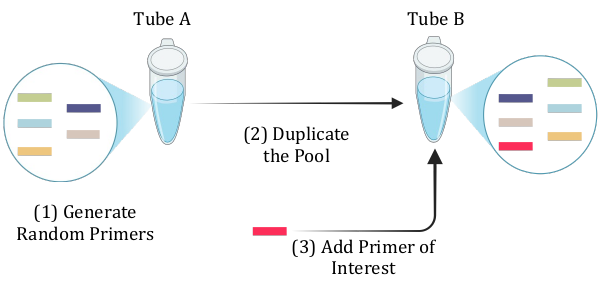}
    \caption{A schematic representation of the user's query preparation.}
    \label{fig:query_prep}
\end{figure}

\begin{figure*}[t]
    \centering
    \includegraphics[width=0.95\linewidth]{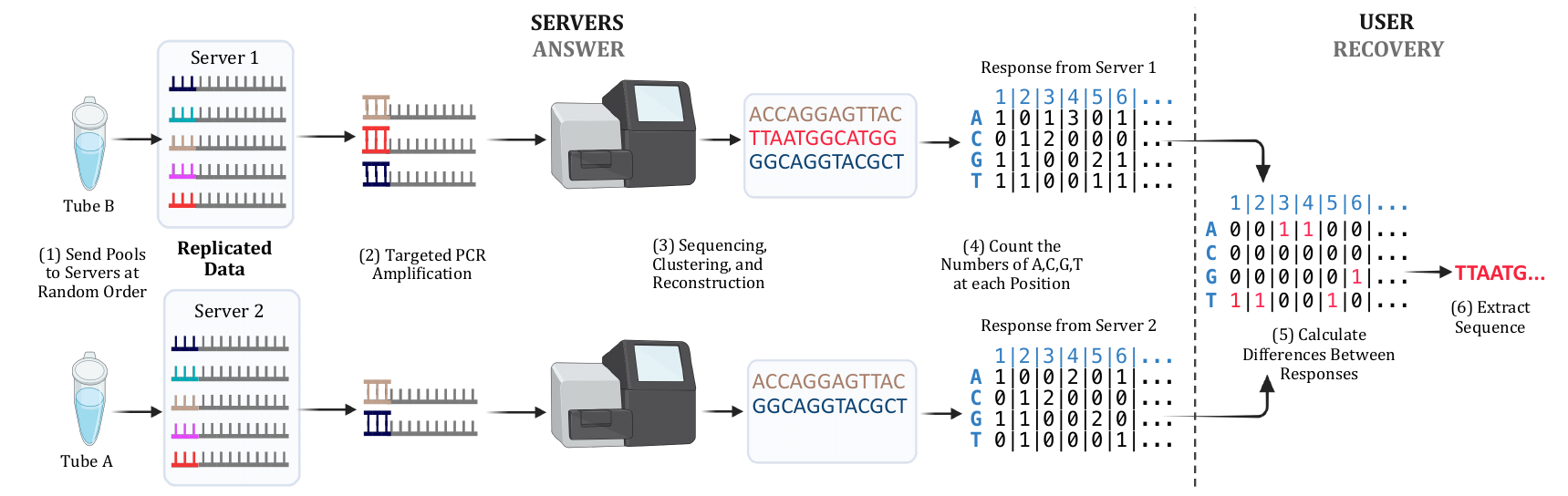}
    \caption{Schematic representations of the subroutines $\algoName{Answer1}$ and $\algoName{Recovery1}$.}
    \label{fig:PIR_response}
\end{figure*}

Upon receiving the queries, both servers perform a biochemical read-out: they add the query primers to their database pools, perform PCR amplification and sequence the output ($\PCRandSequence(\cdot)$). Afterwards, a computational step of clustering and reconstruction ($\Reconstruct(\cdot)$) is performed to obtain a single correct representation\footnote{Here, we assume perfect recovery from $\Reconstruct(\cdot)$, an assumption that holds under sufficient coverage, and is supported by the performance of current state-of-the-art algorithms, such as~\cite{bar2025scalable, rashtchian2017clustering, ben2024gradhc, 
sabary2024reconstruction},
 when combined with appropriate error-correcting codes (e.g.~\cite{bar2025scalable, lenz2021concatenated,
aharoni2025neural,
khabbaz2026dna}).} of all sequences $\rec{i}$ for which  $\barcode{i}\in\barcodePool\cap\queryPool_{j}$, $j \in \set{2}$, whose output is denoted by $\queriedRecords$. The protocols then differ in how the servers computationally process these sequenced reads to generate answers, and how the user computationally reconstructs the queried record.
\begin{algorithm}[ht]
     \label{alg:setup}
    \caption{$\algoName{Setup}(\dbsPool)$}
    $\{\rec{1}, \dots, \rec{\n}\} \leftarrow \dbsPool$\;
    \For{$i \in \set{\n}$}{
        $[\barcode{i}, \payload{i}, \universalPrimer] \leftarrow \rec{i}$\;
    }
    $\barcodePool \leftarrow \{\barcode{1}, \dots, \barcode{\n}\}$\;
    \Return $\barcodePool$
\end{algorithm}
\begin{algorithm}[ht]
    \label{alg:query}
    \caption{$\algoName{Query}(\barcodePool, \p, \f)$}
    $\randomPool \gets \emptyset$\\
    
    \For{$\bfr \in \universalSet{\setminus\{\barcodeOfInterest\}}$}{
        $\randomPool \gets \randomPool\cup\bfr \text{ with probability } \p$\;
    }
    $\randomPoolCopy \leftarrow \Copy(\randomPool)$\;
    $\randomPoolStar \leftarrow \randomPoolCopy \cup \{ \barcodeOfInterest\}$\;
    $(\queryPool_1, \queryPool_2) \getrand \{(\randomPool, \randomPoolStar), (\randomPoolStar, \randomPool)\}$\;
    \Return $\queryPool_1$, $\queryPool_2$
\end{algorithm}
\subsection*{Protocol 1: Base Counter Approach}
After receiving the query, each server in this protocol counts and returns the number of bases $A$, $C$, $G$, $T$ at each position across all sequenced strands, as depicted in Algorithm~\ref{alg:answer1}. 
The answer consists of four numbers for each position $k \in \set{\l}$. The numbers indicate how many times each base has occurred at position $k$ across all sequenced records.
\begin{algorithm}[ht]
    \label{alg:answer1}
    \caption{$\algoName{Answer1}(\dbsPool, \queryPool)$}
    $\ans \leftarrow \{0\}^{4 \times \l}$\;
    $\bfP \leftarrow \PCRandSequence(\dbsPool, \queryPool)$\;
    $\queriedRecords \leftarrow \Reconstruct(\bfP)$\; 
    \For{$\rec{} \in \queriedRecords$}{
        $\ans \leftarrow \ans + \baseToUnit(\rec{})$\;
    }
    \Return $\ans$
\end{algorithm}

\begin{algorithm}[ht]
    \label{alg:recovery1}
    \caption{$\algoName{Recovery1}(\ans_1, \ans_2)$}
    $\ans_{\Delta} \leftarrow |\ans_1 - \ans_2|$\;
    $\hrec{\f} \leftarrow \unitToBase(\ans_{\Delta})$\;
    \Return $\hrec{\f}$
\end{algorithm}

\begin{algorithm}[ht]\label{alg:PIRwithBaseCounter}
    \caption{$\mathsf{PIRwithBaseCounters}(\dbsPool, \p, \f)$}
    $\barcodePool \leftarrow \algoName{Setup}(\dbsPool)$\;
    $\queryPool_1, \queryPool_2 \leftarrow \algoName{Query}(\barcodePool, \p, \f)$\;
    $\ans_1, \ans_2 \leftarrow \algoName{Answer1}(\dbsPool, \queryPool_1), \algoName{Answer1}(\dbsPool, \queryPool_2)$\;
    $\hrec{\f} \leftarrow \algoName{Recovery1}(\ans_1, \ans_2)$\;
    \Return $\hrec{\f}$
\end{algorithm}
After receiving the servers' answers, the user recovers the queried record by taking the difference of the four numbers for each position in the answer matrix (Algorithm~\ref{alg:recovery1}). In each column of the difference matrix, the user will have a 1 for one of the bases and a 0 for the rest. This way, the user can decode the sequence, see Algorithm~\ref{alg:PIRwithBaseCounter} and Fig.~\ref{fig:PIR_response}. 
While intuitive, this protocol incurs a high communication cost since each answer consists of 
$\bigo{\l \log \n}$ bits. Therefore, we use this protocol for presentation purposes only (see Fig.~\ref {fig:PIR_response} for an illustration) and introduce next a more efficient protocol and analyze it.

\subsection*{Protocol 2: XOR Approach}
After receiving the query, each server in this protocol converts the sequenced strands into binary sequences and returns the XOR of the resulting sequences to the user, see Algorithm~\ref{alg:pirwithxor} and Fig.~\ref{fig:xor_response} for an illustration.
\label{subsec:pirwithxor}

The servers use $\baseToBits(\cdot)$ to convert each read sequence into a binary sequence by replacing each base with its corresponding two-bit sequence (Algorithm~\ref{alg:answer2}). To compress the answer, each server then XORs all the sequences and sends back a single sequence, see Algorithm~\ref{alg:recovery2} and Figure~\ref{fig:xor_response}. For the user, downloading two bits for each position $k$, i.e., a total of $2\l$ bits, is much more efficient than downloading $\bigo{\l \log \n}$ bits.
\begin{algorithm}[ht]
    \label{alg:answer2}
    \caption{$\algoName{Answer2}(\dbsPool, \queryPool)$}
    $\ans \leftarrow \{0\}^{2 \times \l}$\;
    $\bfP \leftarrow \PCRandSequence(\dbsPool, \queryPool)$\;
    $\queriedRecords \leftarrow \Reconstruct(\bfP)$\;
    \For{$\rec{} \in \queriedRecords$}{
        $\ans \leftarrow \ans \oplus \baseToBits(\rec{})$\;
    }
    \Return $\ans$
\end{algorithm}

\begin{algorithm}[ht]
    \label{alg:recovery2}
    \caption{$\algoName{Recovery2}(\ans_1, \ans_2)$}
    $\hrec{\f} = \bitsToBase(\ans_1 \oplus \ans_2)$\;
    \Return $\hrec{\f}$
\end{algorithm}
\begin{algorithm}[ht]
    \label{alg:pirwithxor}
    \caption{$\PIRwithXOR(\dbsPool, \p, \f)$}
    $\barcodePool \leftarrow \algoName{Setup}(\dbsPool)$\;
    $\queryPool_1, \queryPool_2 \leftarrow \algoName{Query}(\barcodePool, \p, \f)$\;
    $\ans_1, \ans_2 \leftarrow \algoName{Answer2}(\dbsPool, \queryPool_1), \algoName{Answer2}(\dbsPool, \queryPool_2)$\;
    $\hrec{\f} \leftarrow \algoName{Recovery2}(\ans_1, \ans_2)$\;
    \Return $\hrec{\f}$
\end{algorithm}

After receiving the servers' answers, the user XORs them and obtains the record of interest.  
\begin{figure}[ht]
    \centering
    \includegraphics[width=\colfigwidth{0.7}{0.95}]{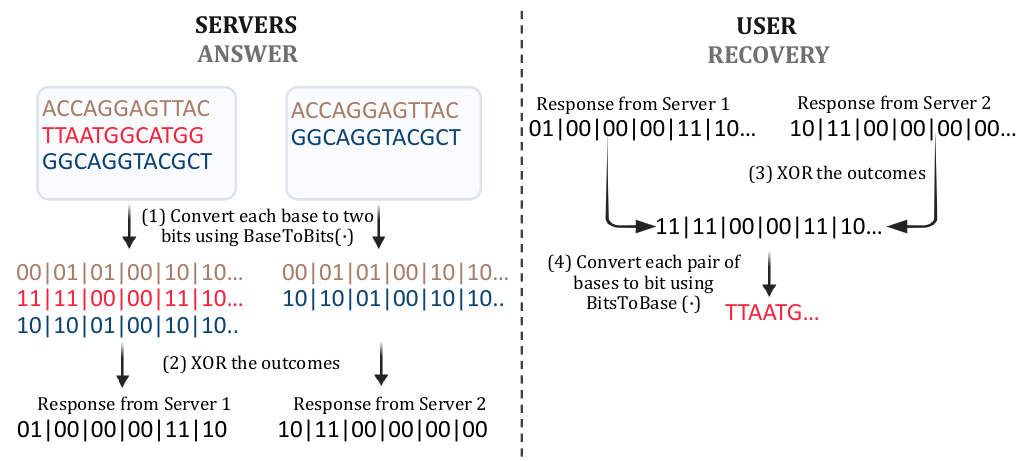}
    \caption{Schematic representations of the subroutines $\algoName{Answer2}$ and $\algoName{Recovery2}$.}
    \label{fig:xor_response}
\end{figure}

\begin{remark}
    A PIR scheme is evaluated by the \emph{PIR rate} defined as the amount of information retrieved divided by the download size. In our $\PIRwithXOR$ protocol, the user retrieves $\rec{\f} \in \baseSet^{\l}$ by downloading a $2 \l$-bit sequence from each of the $\m =2$ servers; hence achieving a
\[
\text{rate} = \frac{2 \l}{\m 2 \l} = \frac{1}{\m}.
\]
\end{remark}

\section{Three PIR Scenarios}
We now instantiate $\PIRwithXOR$ across three scenarios. For each scenario, depending on the reliability of the output of a single instantiation, we utilize $\m/2$ instantiations over $\m/2$ different server pairs for a successful recovery with overwhelming probability.
\begin{remark}
    Algorithms $\algoName{Answer1}$ and $\algoName{Answer2}$ work on the whole records and not just over the payloads. For the first two scenarios below, these algorithms can work solely on the payloads, cutting down communication and computation costs. For the third scenario, as we explain, processing the barcodes is essential for the user to verify that the payload it recovered belongs together with the barcode it initially queried for.
\end{remark}

\subsection{Scenario 1: Checkable Pool, Perfect Copy}
\label{sec:scenario1}

In this scenario, the user can check whether $\barcodeOfInterest \in \randomPool$, and the copy $\randomPoolCopy$ of $\randomPool$ is perfect.
In practice, when synthesizing random sequences, the user cannot avoid sampling $\barcodeOfInterest$ specifically, it is as likely to occur as any other sequence. Yet, we assume that the user can check the contents of $\randomPool$. This way, it can sample a new $\randomPool$ for its queries until $\barcodeOfInterest \notin \randomPool$. Thus, in Algorithm~\ref{alg:query}, the loop is over $\universalSet \setminus \{\barcodeOfInterest\}$. We lift this assumption later in Scenario~2.

In this case, because the user can guarantee that $\randomPoolStar \setminus \randomPool = \{\barcodeOfInterest\}$, the record $\rec{\f}$ can be retrieved from just $\m = 2$ servers.

To analyze the leakage of $\PIRwithXOR$, we start by investigating an easier version of the protocol that fixes the size $\r$ of the random pool. This is intuitive from a technical perspective, but we show that it does not achieve perfect IT privacy.

\begin{theorem}\label{thm:LeakageConstantr} Let $\r \define \size{\randomPool}$ be a constant s.t. $0 \le \r \le \un - 1$. The leakage $\leakage = \mutualInfo{(\queryRV_{j};\indexRV)}$, $j \in [2]$, is given by 
\begin{equation*}
    \leakage = \log \n - \frac{1}{2}\sum_{\k = 0}^{n - 1} \frac{\binom{\n - 1}{\k} \binom{\un - \n}{\r - \k}}{\binom{\un - 1}{\r}} \log ((\k + 1)(\n - \k))
\end{equation*}
and $\leakage \ge 1 - \log(1 + 1 / \n)$ with equality if $\n = 1$.
\end{theorem}
\begin{proof}
    Without loss of generality we give the proof for $\queryRV_1$ and hence drop the subscript. We compute the leakage as 
    \begin{equation}\label{eq:leakage}
    \leakage \define \mutualInfo(\queryRV;\indexRV) = \entropy(\queryRV) - \entropy(\queryRV \mid \indexRV).
    \end{equation}
    To compute the second term, note that given $\indexRV$, the randomness of $\queryRV$ stems from: 
    \begin{enumerate}[label=\emph{(\roman*)}]
        \item receiving $\queryRV = \randomPool$ or $\queryRV = \randomPoolStar$, which happens with probability $1/2$ each; and
        \item the randomness of sampling a subset $\randomPool$ of size $\r$ uniformly at random from $\universalSet\setminus\{\barcodeOfInterest\}$. 
    \end{enumerate}
    Hence, 
    \begin{equation}\label{eq:hqgivenf}
    \entropy(\queryRV \mid \indexRV) = 1 + \log \binom{\un - 1}{\r}.
    \end{equation}
    Now let us analyze $\entropy(\queryRV)$. According to the protocol, $\Pr[\queryPool = \randomPool] = \Pr[\queryPool = \randomPoolStar] = 1/2$. Then, $\forall \randomPool \subset \universalSet$ with $\size{\randomPool} = \r$, $\randomPool$ can be selected if and only if $\barcodeOfInterest \notin \randomPool$. Given $\barcodeOfInterest \notin \randomPool$, the probability of selecting $\randomPool$ is $1 / \binom{\un-1}{\r}$. Any element of $\barcodePool \setminus \randomPool$ can be chosen as the queried index $\barcodeOfInterest$ with probability $1 / \n$, so
    \begin{equation*}
        \Pr[\queryRV = \randomPool] = \frac{1}{2} \cdot \frac{\n - \size{\barcodePool \cap \randomPool}}{\n} \cdot \frac{1}{\binom{\un - 1}{\r}}.
    \end{equation*}
    Similarly, $\forall \randomPoolStar \subseteq \universalSet$ with $\size{\randomPoolStar} = \r + 1$, $\randomPoolStar$ can be obtained if and only if $\barcodeOfInterest \in \randomPoolStar$. Since $\randomPoolStar = \randomPool \cup \{\barcodeOfInterest\}$ for a unique pool $\randomPool \subseteq \universalSet \setminus \{\barcodeOfInterest\}$, any element of $\barcodePool \cap \randomPoolStar$ can be chosen as the queried index $\barcodeOfInterest$ with probability $1/\n$, so
    \begin{align*}
        \Pr[\queryRV = \randomPoolStar] 
        & = \frac{1}{2} \cdot \Pr[\barcodeOfInterest \in \randomPoolStar] \cdot \frac{1}{\binom{\un - 1}{\r}} \\
        & = \frac{1}{2} \cdot \frac{\size{\barcodePool \cap \randomPoolStar}}{\n} \cdot \frac{1}{\binom{\un - 1}{\r}}.
    \end{align*}

    To compute $\entropy(\queryRV)$, we can now iterate over all $\k = \size{\barcodePool \cap \randomPool}$ by choosing $\k$ elements from $\barcodePool$ and $\r - \k$ elements from $\universalSet \setminus \barcodePool$ and over all $\k = \size{\barcodePool \cap \randomPoolStar}$ by choosing $\k$ elements from $\barcodePool$ and $\r - \k + 1$ elements from $\universalSet \setminus \barcodePool$. Then, $\entropy(\queryRV)$ can be written as
    \begin{equation*}
        \begin{split}
            \entropy(\queryRV)
            & = - \sum_{\k = 0}^{\n} \binom{\n}{\k} \binom{\un - \n}{\r - \k} \cdot \frac{\frac{\n - \k}{2\n}}{\binom{\un - 1}{\r}} \cdot \log \left( \frac{\frac{\n - \k}{2\n}}{\binom{\un - 1}{\r}} \right) \\
            & \quad - \sum_{\k = 0}^{\n} \binom{\n}{\k} \binom{\un - \n}{\r - \k + 1} \cdot \frac{\frac{\k}{2\n}}{\binom{\un - 1}{\r}} \cdot \log \left( \frac{\frac{\k}{2\n}}{\binom{\un - 1}{\r}} \right),
        \end{split}
    \end{equation*}
    where terms with $\n - \k = 0$ in the first sum and $\k = 0$ in the second vanish by the convention $0 \log 0 = 0$. We apply the identity $\binom{\n}{\k}(\n - \k) = \n\binom{\n-1}{\k}$ to the first sum, and reindex $\k \mapsto \k - 1$ in the second sum using $\binom{\n}{\k} \k = \n\binom{\n-1}{\k-1}$. This way, the common factor $\binom{\n-1}{\k}\binom{\un-\n}{\r-\k}$ appears in both sums and the upper limit of both becomes $\n - 1$:
    \begin{equation*}
        \begin{split}
            \entropy(\queryRV)
            & = \sum_{\k = 0}^{\n - 1} \frac{\binom{\n - 1}{\k} \binom{\un - \n}{\r - \k}}{2\binom{\un - 1}{\r}} \cdot \left[ \log \binom{\un - 1}{\r} + \log \frac{2\n}{\n - \k}\right] \\
            & \quad + \sum_{\k = 0}^{\n - 1} \frac{\binom{\n - 1}{\k} \binom{\un - \n}{\r - \k}}{2\binom{\un - 1}{\r}} \cdot \left[ \log \binom{\un - 1}{\r} + \log \frac{2\n}{\k + 1}\right].
        \end{split}
    \end{equation*}
    We use Vandermonde's identity as $\sum_{\k=0}^{\n-1}\binom{\n-1}{\k}\binom{\un-\n}{\r-\k} = \binom{\un-1}{\r}$. Combining the two sums and simplifying,
    \begin{equation}\label{eq:hq}
        \begin{split}
            \entropy(\queryRV)
            & = \log \n + \log \binom{\un - 1}{\r} + 1 \\
            & \quad - \frac{1}{2}\sum_{\k = 0}^{\n - 1} \frac{\binom{\n - 1}{\k} \binom{\un - \n}{\r - \k}}{\binom{\un - 1}{\r}} \log ((\k + 1)(\n - \k)).
        \end{split}
    \end{equation}
    From~\cref{eq:hq,eq:hqgivenf,eq:leakage},
    \begin{equation*}
        \leakage = \log \n - \frac{1}{2}\sum_{\k = 0}^{n - 1} \frac{\binom{\n - 1}{\k} \binom{\un - \n}{\r - \k}}{\binom{\un - 1}{\r}} \log ((\k + 1)(\n - \k)).
    \end{equation*}
    Observe for the second term on the RHS that the weights are the hypergeometric probability mass function, so they sum to 1.
    Also, the expression $(\k + 1) + (\n - \k) = \n + 1$ is constant in $\k$, so $(\k+1)(\n - \k) \leq ((\n + 1)/2)^2$ pointwise.
    Then,
    \begin{align*}
        \frac{1}{2}\sum_{\k = 0}^{\n - 1} \frac{\binom{\n - 1}{\k} \binom{\un - \n}{\r - \k}}{\binom{\un - 1}{\r}} \log ((\k + 1)(\n - \k)) 
        & \le \frac{1}{2}\log \left(\frac{\n + 1}{2}\right)^2 \\
        & = \log(\n + 1) - 1
    \end{align*}
    with equality if $\n = 1$. Finally, $\leakage \ge \log \n - \log(n+1) + 1 = 1 - \log((\n+1) / \n) = 1 - \log(1 + 1/\n)$.
\end{proof}
\begin{corollary}
    \label{cor:scenario1fixedrapprox}
    For $\n\r/\un \gg 1$ and $\n(\un - \r)/\un \gg 1$, the leakage in~\cref{thm:LeakageConstantr} is approximately
    \begin{align*}
        \leakage
        & \approx \frac{1}{2} \log \frac{\un^2}{\r(\un - \r)} \\
        & = \frac{1}{2} \KL{\Bern(0)}{\Bern\left(\frac{\r}{\un}\right)} \\
        & \quad + \frac{1}{2} \KL{\Bern(1)}{\Bern\left(\frac{\r}{\un}\right)}.
    \end{align*}
\end{corollary}
\begin{proof}
    The weights of the sum in~\cref{thm:LeakageConstantr} are the hypergeometric probability mass function of $\k = \size{\barcodePool \cap \randomPool}$, where $\randomPool$ is sampled from $\universalSet \setminus \{\barcodeOfInterest\}$ and $\barcodeOfInterest \notin \randomPool$. For large $\n$, $\k$ concentrates around its mean $(\n - 1)\r/(\un - 1) \approx \n\r/\un$. Setting $\k = \n\r/\un$ and approximating $(\k + 1)(\n - \k)$ by $\k(\n - \k)$,
    \begin{align*}
        \leakage
        & \approx \log \n - \frac{1}{2} \log \left(\frac{\n\r}{\un} \cdot \frac{\n(\un - \r)}{\un}\right) \\
        & = \frac{1}{2} \log \frac{\un^2}{\r(\un - \r)}.
    \end{align*}
    The KL divergence form follows from 
    \[
    \KL{\Bern(0)}{\Bern\left(\frac{\r}{\un}\right)} = \log \frac{\un}{\un - \r}
    \]
    and
    \[
    \KL{\Bern(1)}{\Bern\left(\frac{\r}{\un}\right)} = \log \frac{\un}{\r}.
    \]
\end{proof}
By observing $\size{\queryPool}$, a server knows whether it received $\randomPool$ or $\randomPoolStar$
If it received $\randomPool$, $\barcodeOfInterest$ is not in its query, and if it received $\randomPoolStar$, $\barcodeOfInterest$ is in its query, while any other barcode is in its query with probability $\r/\un$ in both cases. \cref{cor:scenario1fixedrapprox} averages the KL divergences of these two cases. Since $\r(\un - \r) \le \un^2/4$, the approximation is at least $1$ bit.


The leakage for $\lbc = 20$ is plotted in Fig.~\ref{fig:informationleakagefixedr}. Given $\queryPool$, the server can count $q\define |\queryPool|$ and deduce whether $\queryPool=\randomPool$ if $q=\r$ or $\queryPool=\randomPoolStar$ otherwise. For small and large $\r$, the server knows with $50\%$ chance that $\f$ is from a small set, and the leakage is high. For $\k = \size{\barcodePool \cap \randomPool} = \lfloor \n /2 \rfloor$, the server knows either $\f$ is among a set of size $\n - \lfloor \n / 2 \rfloor$, or $\f$ is among a set of size $\lfloor \n / 2 \rfloor + 1$, leaking roughly 1 bit of information.

Theorem~\ref{thm:LeakageConstantr} shows that the fixed $\r$ case cannot achieve $0$ leakage for any non-trivial database size of $\n > 1$. Hence, in~\Cref{thm:scenario1infoleak} we analyze the leakage for $\PIRwithXOR$ where $\r = \size{\randomPool} \leftarrow \Bin(\un - 1, \p)$, as for any $\bfr \in \universalSet \setminus \{\barcodeOfInterest\}$, $\bfr \in \randomPool$ with probability $\p$. The leakage is plotted in Fig.~\ref{fig:informationleakagefixedp}. In this case, the server can only deduce probabilistically whether its query includes the queried index. For $\p = 1/2$, this probability is $1/2$ and $\PIRwithXOR$ achieves perfect IT privacy.
\begin{theorem}
    \label{thm:scenario1infoleak}
    Let $\randomPool \subset \universalSet$ be a random pool of sequences and for any sequence $\bfr \in \universalSet \setminus \{\barcodeOfInterest\}$, let $\Pr[\bfr \in \randomPool] = \p$ for $0 < \p < 1$. Then the leakage $\leakage = \mutualInfo{(\queryRV_{j};\indexRV)}$, $j \in [2]$ is given by
    \begin{multline*}
        \leakage = \log \n + \frac{1}{2} \log \p (1 - \p)
        - \frac{1}{2\n \p (1 - \p)} \\
        \cdot \sum_{\q = 0}^{\un} \binom{\un}{\q} \p^{\q} (1 - \p)^{\un - \q} \sum_{\k = 0}^{\q} \frac{\binom{\n}{\k}\binom{\un - \n}{\q - \k}}{\binom{\un}{\q}} \kterm \log \kterm,
    \end{multline*}
    where $\kterm := \n \p + \k(1 - 2\p)$.
\end{theorem}
\begin{proof}
    Given $\indexRV$, the randomness of $\queryRV$ stems from:
    \begin{enumerate}[label=\emph{(\roman*)}]
    \item $\un - 1$ Bernoulli random variables with success probability $\p$,
    \item and the uniformly random choice between $\randomPool$ and $\randomPoolStar$.
    \end{enumerate}
    Thus
    \[
        \entropy(\queryRV \mid \indexRV) = (\un - 1) \binaryEntropy(\p) + 1.
    \]
    Now we analyze $\entropy(\queryRV)$. For $\q = \size{\queryPool}$, $\k = \size{\barcodePool \cap \queryPool}$, we can write
    \begin{align*}
        \Pr[\queryRV = \queryPool] 
        & = \frac{1}{2}\Pr[\barcodeOfInterest \in \queryPool] \p^{\q - 1} (1 - \p)^{\un - \q} \\
        & \quad + \frac{1}{2}\Pr[\barcodeOfInterest \notin \queryPool] \p^{\q} (1 - \p)^{\un - \q - 1} \\
        & = \frac{1}{2\n} \p^{\q - 1} (1 - \p)^{\un - \q - 1} (\n\p + \k(1 - 2\p))
    \end{align*}
    since $\Pr[\barcodeOfInterest \in \queryPool] = \k / \n$. We can now compute $\entropy(\queryRV)$ as
    \begin{equation*}
        \entropy(\queryRV) 
        = - \sum_{\q = 0}^{\un} \sum_{\k = 0}^{\q} \binom{\n}{\k}\binom{\un - \n}{\q - \k} \Pr[\queryRV = \queryPool] \log \Pr[\queryRV = \queryPool].
    \end{equation*}
    The sum above can be split into four parts since
    \begin{multline*}
        \log \Pr[\queryRV = \queryPool] = - \log 2\n + (\q - 1) \log \p \\ 
        + (\un - \q - 1) \log (1 - \p) + \log(\n\p + \k(1 - 2\p)).
    \end{multline*}
    This way, we can write $\entropy(\queryRV) = \entropy_1 + \entropy_2 + \entropy_3 + \entropy_4$ and calculate each part individually. For $\entropy_1$, $\entropy_2$ and $\entropy_3$ below, we utilize Vandermonde's identity as
    \begin{equation*}
        \sum_{\k = 0}^{\q} \binom{\n}{\k} \binom{\un - \n}{\q - \k} \kterm = \n\p \binom{\un}{\q} + \n(1 - 2\p)\binom{\un - 1}{\q - 1}.
    \end{equation*}
    \begin{align*}
        \entropy_1 
        & = \log2\n \sum_{\q = 0}^{\un} \sum_{\k = 0}^{\q} \binom{\n}{\k}\binom{\un - \n}{\q - \k} \\
        & \quad \cdot \frac{1}{2\n} \p^{\q - 1} (1 - \p)^{\un - \q - 1} (\n\p + \k(1 - 2\p)) \\
        & = \frac{\log2\n}{2\n} \sum_{\q = 0}^{\un} \p^{\q - 1} (1 - \p)^{\un - \q - 1} \\
        & \quad \cdot \sum_{\k = 0}^{\q} \binom{\n}{\k}\binom{\un - \n}{\q - \k} (\n\p + \k(1 - 2\p)) \\
        & = \frac{\log2\n}{2} \sum_{\q = 0}^{\un} \p^{\q - 1} (1 - \p)^{\un - \q - 1} \\
        & \quad \cdot \left(\p \binom{\un}{\q} + (1 -2\p) \binom{\un - 1}{\q - 1} \right) \\
        & = \frac{\log2\n}{2(1 - \p)} + \frac{\log2\n}{2\p(1 - \p)} \frac{1 - 2\p}{\un} \un \p = \log 2\n.
    \end{align*}
    \begin{align*}
        \entropy_2
        & = - \log \p \sum_{\q = 0}^{\un} \sum_{\k = 0}^{\q} \binom{\n}{\k}\binom{\un - \n}{\q - \k} \\
        & \quad \cdot \frac{1}{2\n} \p^{\q - 1} (1 - \p)^{\un - \q - 1} (\n\p + \k(1 - 2\p)) (\q - 1) \\
        & = - \frac{\log \p}{2} \sum_{\q = 0}^{\un} (\q - 1) \p^{\q - 1} (1 - \p)^{\un - \q - 1} \\
        & \quad \cdot \left(\p \binom{\un}{\q} + (1 -2\p) \binom{\un - 1}{\q - 1} \right) \\
        & = -\frac{\log\p}{2(1 - \p)} (\un\p - 1 + (\un - 1)\p(1 - 2\p)) \\
        & = -\left((\un - 1)\p - \tfrac{1}{2}\right)\log\p.
    \end{align*}
    \begin{align*}
        \entropy_3
        & = - \log (1 - \p) \sum_{\q = 0}^{\un} \sum_{\k = 0}^{\q} \binom{\n}{\k}\binom{\un - \n}{\q - \k} \\
        & \quad \cdot \frac{1}{2\n} \p^{\q - 1} (1 - \p)^{\un - \q - 1} (\n\p + \k(1 - 2\p)) (\un - \q - 1) \\
        & = - \frac{\log (1 - \p)}{2} \sum_{\q = 0}^{\un} (\un - \q - 1) \p^{\q - 1} (1 - \p)^{\un - \q - 1} \\
        & \quad \cdot \left(\p \binom{\un}{\q} + (1 -2\p) \binom{\un - 1}{\q - 1} \right) \\
        & = - \frac{\log(1 - \p)}{2}(2\un(1 - \p) + 2\p - 3) \\
        & = - \left((\un - 1)(1 - \p) - \tfrac{1}{2}\right) \log(1 - \p). 
    \end{align*}
    \begin{align*}
        \entropy_4 
        & = - \frac{1}{2\n} \sum_{\q = 0}^{\un} \sum_{\k = 0}^{\q} \binom{\n}{\k}\binom{\un - \n}{\q - \k} \p^{\q - 1} (1 - \p)^{\un - \q - 1} \\
        & \quad \cdot (\n\p + \k(1 - 2\p)) \log(\n\p + \k(1 - 2\p)) \\
        & = - \frac{1}{2\n \p (1 - \p)} \sum_{\q = 0}^{\un} \binom{\un}{\q} \p^{\q} (1 - \p)^{\un - \q} \\
        & \quad \cdot \sum_{\k = 0}^{\q} \frac{\binom{\n}{\k}\binom{\un - \n}{\q - \k}}{\binom{\un}{\q}} (\n\p + \k(1 - 2\p)) \log(\n\p + \k(1 - 2\p)).
    \end{align*}
    Then the mutual information can be computed as 
    \[
    \mutualInfo(\queryRV; \indexRV) = \entropy_1 + \entropy_2 + \entropy_3 + \entropy_4 - \entropy(\queryRV \mid \indexRV).
    \]
    Combining $\entropy_2$ and $\entropy_3$,
    \begin{align*}
        \entropy_2 + \entropy_3
        & = -\left((\un - 1)\p - \tfrac{1}{2}\right)\log \p \\
        & \quad - \left((\un - 1)(1 - \p) - \tfrac{1}{2}\right)\log(1 - \p) \\
        & = (\un - 1)\binaryEntropy(\p) + \tfrac{1}{2}\log \p(1 - \p).
    \end{align*}
    Since $\entropy_1 = \log 2\n$ and $\entropy(\queryRV \mid \indexRV) = (\un - 1)\binaryEntropy(\p) + 1$,
    \begin{align*}
        \leakage
        & = \entropy_1 + \entropy_2 + \entropy_3 + \entropy_4 - \entropy(\queryRV \mid \indexRV) \\
        & = \log 2\n + (\un - 1)\binaryEntropy(\p) + \tfrac{1}{2}\log \p(1 - \p) \\
        & \quad + \entropy_4 - (\un - 1)\binaryEntropy(\p) - 1 \\
        & = \log \n + \tfrac{1}{2}\log \p(1 - \p) + \entropy_4,
    \end{align*}
    and the result follows.
\end{proof}
\begin{corollary}
    For $\p = 1/2$, $\leakage = 0$.    
\end{corollary}
\begin{proof}
    Follows by setting $\p = 1/2$ in Theorem~\ref{thm:scenario1infoleak}. Alternatively, let us first analyze $\entropy(\queryRV \mid \indexRV = \f)$. For any subset $\randomPool \subseteq \universalSet \setminus \{\barcodeOfInterest\}$,
    \begin{align*}
        \Pr[\randomPoolRV = \randomPool \mid \indexRV]
        & = \Pr[\size{\randomPool} = \r \mid \indexRV] \cdot \Pr[\randomPoolRV = \randomPool \mid \indexRV, \,\size{\randomPool} = \r] \\
        & = \frac{\binom{\un - 1}{\r}}{2^{\un - 1}} \cdot \frac{1}{\binom{\un - 1}{\r}} = \frac{1}{2^{\un-1}}.
    \end{align*}
    Thus, $\entropy(\randomPoolRV \mid \indexRV) = \un - 1$. The proof follows by $\entropy(\queryRV \mid \indexRV) = \entropy(\randomPoolRV \mid \indexRV) + 1$ as $\un \geq \entropy(\queryRV) \geq \entropy(\queryRV \mid \indexRV) = \un$.
\end{proof}

\begin{corollary}
    \label{cor:scenario1approx}
    For $\n\p \gg 1$ and $\n(1 - \p) \gg 1$, the leakage in~\cref{thm:scenario1infoleak} is approximately
    \begin{equation*}
        \leakage \approx -1 - \frac{1}{2} \log \p(1 - \p) = \KL{\Bern\left(\frac{1}{2}\right)}{\Bern(\p)}.
    \end{equation*}
\end{corollary}
\begin{proof}
    The weights of the double sum in~\cref{thm:scenario1infoleak} are the probability that $\size{\barcodePool \cap \queryPool} = \k$ for a query pool $\queryPool$ that contains every sequence in $\universalSet$ with probability $\p$, i.e., $\k \sim \Bin(\n, \p)$. For large $\n$, $\k$ concentrates around $\n\p$. Setting $\k = \n\p$, we get $\kterm = 2\n\p(1 - \p)$, and since the weights sum up to 1,
    \begin{align*}
        \leakage
        & \approx \log \n + \frac{1}{2} \log \p(1 - \p) - \log 2\n\p(1 - \p) \\
        & = -1 - \frac{1}{2} \log \p(1 - \p).
    \end{align*}
    The KL divergence form follows from 
    \[
    \KL{\Bern(1/2)}{\Bern(\p)} = \frac{1}{2} \log \frac{1}{2\p} + \frac{1}{2} \log \frac{1}{2(1 - \p)}. \qedhere
    \]
\end{proof}
\cref{cor:scenario1approx} explains the leakage: a server compares how often $\barcodeOfInterest$ is in its query with how often any other barcode is in its query. The former happens with probability $1/2$ and the latter happens with probability $\p$. The two coincide for $\p = 1/2$.
\begin{figure}[ht]
    \centering
    \includegraphics[width=\colfigwidth{0.8}{1}]{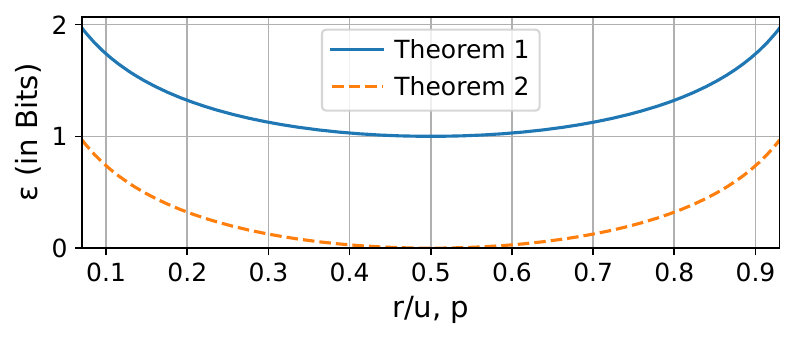}
    \caption{Leakage in Scenario 1 for fixed $\r$ (\cref{thm:LeakageConstantr}) and fixed $\p$ (\cref{thm:scenario1infoleak}). $\un \approx 1.1 \cdot 10^{12}, \n \approx 2.7 \cdot 10^{9}$. $\entropy(\indexRV) = 31.36$. The x-axis range is restricted to where the leakage is smaller than 1 bit for~\Cref{thm:scenario1infoleak}.}
    \label{fig:informationleakagefixedr}
    \label{fig:informationleakagefixedp}
\end{figure}
\subsection{Scenario 2: Uncheckable Pool, Perfect Copy}
\label{sec:scenario2}
In this scenario, the user is unable to check whether $\barcodeOfInterest \in \randomPool$, and the copy $\randomPoolCopy$ of $\randomPool$ is perfect. This means that $\barcodeOfInterest$ appears in $\randomPool$ with probability $\p$. This would mean $\randomPool = \randomPoolStar \rightarrow \queryPool_1 = \queryPool_2 \rightarrow \ans_1 = \ans_2$, and the retrieval fails. To ensure a low failure probability, we repeat the protocol $\m/2$ times over $\m/2$ server pairs. We start by analyzing the easier case of fixed $\r = \size{\randomPool}$, and continue with the analysis of our actual protocol.
\begin{theorem}
    Let $\r \define \size{\randomPool}$ be a constant s.t. $0 \le \r \le \un$. For a given $\lambda>0$, the failure probability of our protocol naively repeated over $\m \ge 2\lambda / (2\lbc - \log \r)$ servers is at most $2^{-\lambda}$. For any $j \in [2]$, the leakage $\leakage = \mutualInfo{(\queryRV_{j};\indexRV)}$ is 
    \begin{equation*}
        \begin{split}
            \leakage
            & = \log \n - \frac{1}{\binom{\un}{\r}} \sum_{\k = 1}^{\n} \bigg[ \binom{\n - 1}{\k - 1}\binom{\un - \n}{\r - \k} \log \k \\
            & \quad + \frac{1}{2} \binom{\n - 1}{\k - 1}\binom{\un - \n}{\r - \k + 1} \log (\k (\n - \k + 1)) \\
            & \quad + \frac{\n + \k}{2\n} \binom{\n}{\k}\binom{\un - \n}{\r - \k} \binaryEntropy\left(\frac{2\k}{\n + \k}\right) \bigg].
        \end{split}
    \end{equation*}
    In particular, $\leakage = (\log \n)/2$ for $\r = 0$ and $\leakage = 0$ for $\r = \un$.
    \label{thm:scenario2constantr}
\end{theorem}
\begin{proof}
    The user will succeed in retrieving the queried record if at least one of the $(\randomPool, \randomPoolStar)$ pairs is different. This is due to the fact that $\randomPool \neq \randomPoolStar$ guarantees $\randomPoolStar \setminus \randomPool = \{\barcodeOfInterest\}$ since $\randomPool = \randomPoolCopy$. So,
    \begin{equation*}
        \Pr[\text{failure}] = \Pr[\barcodeOfInterest \in \randomPool]^{\m / 2}.
    \end{equation*}
    Thus, we just need to find the probability of $\barcodeOfInterest$ being in $\randomPool$. Because the user samples $\randomPool$ of size $\size{\randomPool} = \r$ from a universal pool of size $\size{\universalSet} = \un$, $\Pr[\barcodeOfInterest \in \randomPool] = \r/\un$. For a negligible failure probability, we write
    \begin{align*}
        2^{-\lambda} 
        & \geq \Pr[\barcodeOfInterest \in \randomPool]^{\m / 2} \\
        & = \left(\frac{\r}{\un}\right)^{\m / 2}
    \end{align*}
    and obtain $\m \ge 2\lambda / (2\lbc - \log \r)$ since $\un = 4^{\lbc}$. 

    For the leakage, we start with
    \begin{align}\label{eq:s2r_leakage}
        \mutualInfo{(\queryRV; \indexRV)}
        & = \entropy(\indexRV) - \entropy(\indexRV \mid \queryRV) \nonumber \\
        & = \log \n - \entropy(\indexRV \mid \queryRV),
    \end{align}
    where $\entropy(\indexRV) = \log \n$ is due to $\indexRV$ being uniform and $\size{\barcodePool} = \n$. For a fixed $\r = \size{\randomPool}$, we can break down the second term as
    \begin{equation}\label{eq:s2r_split}
        \begin{split}
            \entropy(\indexRV \mid \queryRV)
            & = \Pr[\size{\queryRV} = \r] \cdot \entropy(\indexRV \mid \queryRV, \size{\queryRV} = \r) \\
            & \quad + \Pr[\size{\queryRV} = \r + 1] \cdot \entropy(\indexRV \mid \queryRV, \size{\queryRV} = \r  + 1).
        \end{split}
    \end{equation}
    The event $\size{\queryRV} = \r$ happens in two ways:
    \begin{enumerate}
        \item if $\barcodeOfInterest \in \randomPool$. Then $\size{\queryRV} = \r$ regardless of whether the server receives $\randomPool$ or $\randomPoolStar$.
        \item if $\barcodeOfInterest \notin \randomPool$, and the server receives $\queryRV = \randomPool$.
    \end{enumerate}
    Then
    \begin{align}\label{eq:s2r_prob_r}
        \Pr[\size{\queryRV} = \r] 
        & = \Pr[\barcodeOfInterest \in \randomPool] + \Pr[\barcodeOfInterest \notin \randomPool] \cdot \Pr[\queryRV = \randomPool] \nonumber \\
        & = \frac{\r}{\un} + \frac{\un - \r}{2\un} = \frac{\un + \r}{2\un}.
    \end{align}
    When a server observes $\size{\queryRV} = \r$, it can infer that
    \begin{equation}\label{eq:s2r_in_given_r}
    \Pr[\barcodeOfInterest \in \randomPool \mid \size{\queryRV} = \r] = \frac{\r}{\un} \frac{2\un}{\un + \r} = \frac{2\r}{\un + \r}
    \end{equation}
    and 
    \begin{equation}\label{eq:s2r_notin_given_r}
    \Pr[\barcodeOfInterest \notin \randomPool \mid \size{\queryRV} = \r] = \frac{\un - \r}{2\un} \frac{2\un}{\un + \r} = \frac{\un - \r}{\un + \r}.
    \end{equation}
    However, the server does not know which of the two cases holds. Since by knowing $\queryRV = \queryPool$ alone without $\indexRV$, the server cannot know whether $\barcodeOfInterest \in \randomPool$, this introduces an entropy term $\entropy(\indicator{\barcodeOfInterest \in \randomPool} \mid \queryRV, \size{\queryRV} = \r)$. Then we can describe $\entropy(\indexRV \mid \queryRV, \size{\queryRV} = \r)$ as
    \begin{equation}\label{eq:s2r_split_r}
        \begin{split}
            \entropy(\indexRV \mid \queryRV, \size{\queryRV} = \r)
            & = \Pr[\barcodeOfInterest \in \randomPool \mid \size{\queryRV} = \r] \\
            & \quad \cdot \entropy(\indexRV \mid \queryRV, \size{\queryRV} = \r, \barcodeOfInterest \in \randomPool) \\
            & \quad + \Pr[\barcodeOfInterest \notin \randomPool \mid \size{\queryRV} = \r] \\
            & \quad \cdot \entropy(\indexRV \mid \queryRV, \size{\queryRV} = \r, \barcodeOfInterest \notin \randomPool) \\
            & \quad + \entropy(\indicator{\barcodeOfInterest \in \randomPool} \mid \queryRV, \size{\queryRV} = \r).
        \end{split}
    \end{equation}
    Let us define $\k \define \size{\barcodePool \cap \queryPool}$ for a query $\queryPool$. Given $\size{\queryRV} = \r$ and $\barcodeOfInterest \in \randomPool$, to satisfy $\k = \size{\barcodePool \cap \queryPool}$, $\k - 1$ elements of the intersection have to be chosen among $\n - 1$ elements of $\barcodePool$. The remaining $\r - \k$ elements have to be chosen from $\universalSet \setminus \barcodePool$. With no regard to the intersection size, $\r - 1$ elements have to be chosen from $\universalSet \setminus \{\barcodeOfInterest\}$. Any of the $\k$ elements can be the barcode of interest with equal probability, thus we can write
    \begin{equation}\label{eq:s2r_entropy_in}
        \entropy(\indexRV \mid \queryRV, \size{\queryRV} = \r, \barcodeOfInterest \in \randomPool) = \sum_{\k = 1}^\n \frac{\binom{\n - 1}{\k - 1}\binom{\un - \n}{\r - \k}}{\binom{\un - 1}{\r - 1}} \log \k.
    \end{equation}
    Then, given $\size{\queryRV} = \r$ and $\barcodeOfInterest \notin \randomPool$, to satisfy $\k = \size{\barcodePool \cap \queryPool}$, $\k$ elements have to be chosen from $\barcodePool \setminus \{\barcodeOfInterest\}$ and the remaining $\r - \k$ elements have to be chosen from $\universalSet \setminus \barcodePool$. With no regard to the intersection size, $\r$ elements have to be chosen from $\universalSet \setminus \{\barcodeOfInterest\}$. Any of the $\n - \k$ elements of $\barcodePool \setminus \queryPool$ can be the barcode of interest with equal probability, so we can write
    \begin{equation}\label{eq:s2r_entropy_notin}
        \entropy(\indexRV \mid \queryRV, \size{\queryRV} = \r, \barcodeOfInterest \notin \randomPool)
        = \sum_{\k = 0}^{\n - 1} \frac{\binom{\n - 1}{\k}\binom{\un - \n}{\r - \k}}{\binom{\un - 1}{\r}} \log (\n - \k).
    \end{equation}
    For the last term of~\cref{eq:s2r_split_r}, consider a query $\queryPool$ with $\size{\queryPool} = \r$ and $\k = \size{\barcodePool \cap \queryPool}$. The barcode of interest is one of the $\k$ elements of $\barcodePool \cap \queryPool$ with probability $\k / \n$, since $\randomPool$ is sampled independently of $\indexRV$. Repeating the argument of~\cref{eq:s2r_in_given_r} with $\k / \n$ in place of $\r / \un$, we get $\Pr[\barcodeOfInterest \in \randomPool \mid \queryRV = \queryPool] = 2\k / (\n + \k)$, and the entropy of $\indicator{\barcodeOfInterest \in \randomPool}$ given $\queryRV = \queryPool$ is $\binaryEntropy(2\k / (\n + \k))$. Weighting the distributions of $\k$ in~\cref{eq:s2r_entropy_in,eq:s2r_entropy_notin} with~\cref{eq:s2r_in_given_r,eq:s2r_notin_given_r}, we can write
    \begin{equation}\label{eq:s2r_entropy_ind}
        \begin{split}
            & \entropy(\indicator{\barcodeOfInterest \in \randomPool} \mid \queryRV, \size{\queryRV} = \r) \\
            & = \frac{2\r}{\un + \r} \sum_{\k = 1}^{\n} \frac{\binom{\n - 1}{\k - 1}\binom{\un - \n}{\r - \k}}{\binom{\un - 1}{\r - 1}} \binaryEntropy\left(\frac{2\k}{\n + \k}\right) \\
            & \quad + \frac{\un - \r}{\un + \r} \sum_{\k = 0}^{\n - 1} \frac{\binom{\n - 1}{\k}\binom{\un - \n}{\r - \k}}{\binom{\un - 1}{\r}} \binaryEntropy\left(\frac{2\k}{\n + \k}\right).
        \end{split}
    \end{equation}
    Plugging~\cref{eq:s2r_in_given_r,eq:s2r_notin_given_r,eq:s2r_entropy_in,eq:s2r_entropy_notin,eq:s2r_entropy_ind} into~\cref{eq:s2r_split_r}, we obtain
    \begin{equation}\label{eq:s2r_entropy_r}
        \begin{split}
            & \entropy(\indexRV \mid \queryRV, \size{\queryRV} = \r) \\
            & = \frac{2\r}{\un + \r} \sum_{\k = 1}^{\n} \frac{\binom{\n - 1}{\k - 1}\binom{\un - \n}{\r - \k}}{\binom{\un - 1}{\r - 1}} \left[\log \k + \binaryEntropy\left(\frac{2\k}{\n + \k}\right)\right] \\
            & \quad + \frac{\un - \r}{\un + \r} \sum_{\k = 0}^{\n - 1} \frac{\binom{\n - 1}{\k}\binom{\un - \n}{\r - \k}}{\binom{\un - 1}{\r}} \left[\log (\n - \k) + \binaryEntropy\left(\frac{2\k}{\n + \k}\right)\right].
        \end{split}
    \end{equation}
    For the event $\size{\queryRV} = \r + 1$, we compute
    \begin{align}\label{eq:s2r_prob_r1}
        \Pr[\size{\queryRV} = \r + 1]
        & = 1 - \Pr[\size{\queryRV} = \r] \nonumber \\
        & = \Pr[\barcodeOfInterest \notin \randomPool] \cdot \Pr[\queryRV = \randomPoolStar] \nonumber \\
        & = \frac{\un -\r}{2\un}.
    \end{align}
    In this event, the server knows that $\size{\barcodePool \cap \queryPool} \ge 1$ due to $\barcodeOfInterest \in \queryPool$. As $\size{\barcodePool \cap \queryPool} = \k$, $\k - 1$ elements have to be chosen from $\barcodePool \setminus \{\barcodeOfInterest\}$, and $\r - \k + 1$ elements have to be chosen from $\universalSet \setminus \barcodePool$. With no regard to the intersection size, $\r$ elements have to be chosen from $\universalSet \setminus \{\barcodeOfInterest\}$. Any of the $\k$ elements can be the barcode of interest with equal probability, so we write
    \begin{equation}\label{eq:s2r_entropy_r1}
    \entropy(\indexRV \mid \queryRV, \size{\queryRV} = \r + 1) = \sum_{\k = 1}^{\n} \frac{\binom{\n - 1}{\k - 1}\binom{\un - \n}{\r - \k + 1}}{\binom{\un - 1}{\r}}\log \k.
    \end{equation}
    Bringing it all together, we plug~\cref{eq:s2r_prob_r,eq:s2r_entropy_r,eq:s2r_prob_r1,eq:s2r_entropy_r1} into~\cref{eq:s2r_split} and obtain
    \begin{equation}\label{eq:s2r_total_plugged}
        \begin{split}
            & \entropy(\indexRV \mid \queryRV) \\
            & = \frac{\r}{\un} \sum_{\k = 1}^{\n} \frac{\binom{\n - 1}{\k - 1}\binom{\un - \n}{\r - \k}}{\binom{\un - 1}{\r - 1}} \left[\log \k + \binaryEntropy\left(\frac{2\k}{\n + \k}\right)\right] \\
            & \quad + \frac{\un - \r}{2\un} \sum_{\k = 0}^{\n - 1} \frac{\binom{\n - 1}{\k}\binom{\un - \n}{\r - \k}}{\binom{\un - 1}{\r}} \left[\log (\n - \k) + \binaryEntropy\left(\frac{2\k}{\n + \k}\right)\right] \\
            & \quad + \frac{\un - \r}{2\un} \sum_{\k = 1}^{\n} \frac{\binom{\n - 1}{\k - 1}\binom{\un - \n}{\r - \k + 1}}{\binom{\un - 1}{\r}} \log \k.
        \end{split}
    \end{equation}
    First, we simplify the factors in front of the sums. Since $(\r / \un) \binom{\un}{\r} = \binom{\un - 1}{\r - 1}$ and $((\un - \r) / \un) \binom{\un}{\r} = \binom{\un - 1}{\r}$, \cref{eq:s2r_total_plugged} becomes
    \begin{equation}\label{eq:s2r_total_factors}
        \begin{split}
            & \binom{\un}{\r} \cdot \entropy(\indexRV \mid \queryRV) \\
            & = \sum_{\k = 1}^{\n} \binom{\n - 1}{\k - 1}\binom{\un - \n}{\r - \k} \left[\log \k + \binaryEntropy\left(\frac{2\k}{\n + \k}\right)\right] \\
            & \quad + \frac{1}{2} \sum_{\k = 0}^{\n - 1} \binom{\n - 1}{\k}\binom{\un - \n}{\r - \k} \left[\log (\n - \k) + \binaryEntropy\left(\frac{2\k}{\n + \k}\right)\right] \\
            & \quad + \frac{1}{2} \sum_{\k = 1}^{\n} \binom{\n - 1}{\k - 1}\binom{\un - \n}{\r - \k + 1} \log \k.
        \end{split}
    \end{equation}
    Next, we bring all sums to the range $\k = 1, \dots, \n$. In the second sum of~\cref{eq:s2r_total_factors}, we reindex $\k \mapsto \k - 1$ in the part with $\log (\n - \k)$ to get
    \begin{equation*}
        \begin{split}
            & \sum_{\k = 0}^{\n - 1} \binom{\n - 1}{\k}\binom{\un - \n}{\r - \k} \log (\n - \k) \\
            & = \sum_{\k = 1}^{\n} \binom{\n - 1}{\k - 1}\binom{\un - \n}{\r - \k + 1} \log (\n - \k + 1),
        \end{split}
    \end{equation*}
    which has the same binomial coefficients as the third sum, so the two can be merged using $\log \k + \log (\n - \k + 1) = \log (\k (\n - \k + 1))$. In the part with the binary entropy, the term with $\k = 0$ vanishes due to $\binaryEntropy(0) = 0$, and we can add the term with $\k = \n$ due to $\binom{\n - 1}{\n} = 0$. Then,
    \begin{equation}\label{eq:s2r_total_reindexed}
        \begin{split}
            & \binom{\un}{\r} \cdot \entropy(\indexRV \mid \queryRV) \\
            & = \sum_{\k = 1}^{\n} \bigg[ \binom{\n - 1}{\k - 1}\binom{\un - \n}{\r - \k} \log \k \\
            & \quad + \frac{1}{2} \binom{\n - 1}{\k - 1}\binom{\un - \n}{\r - \k + 1} \log (\k (\n - \k + 1)) \\
            & \quad + \left(\binom{\n - 1}{\k - 1} + \frac{1}{2}\binom{\n - 1}{\k}\right) \binom{\un - \n}{\r - \k} \binaryEntropy\left(\frac{2\k}{\n + \k}\right) \bigg].
        \end{split}
    \end{equation}
    Finally, we use $\binom{\n - 1}{\k - 1} = (\k / \n) \binom{\n}{\k}$ and $\binom{\n - 1}{\k} = ((\n - \k) / \n) \binom{\n}{\k}$ to write
    \begin{equation*}
        \binom{\n - 1}{\k - 1} + \frac{1}{2}\binom{\n - 1}{\k} = \frac{\n + \k}{2\n} \binom{\n}{\k},
    \end{equation*}
    and~\cref{eq:s2r_total_reindexed} becomes
    \begin{equation}\label{eq:s2r_entropy_total}
        \begin{split}
            \entropy(\indexRV \mid \queryRV)
            & = \frac{1}{\binom{\un}{\r}} \sum_{\k = 1}^{\n} \bigg[ \binom{\n - 1}{\k - 1}\binom{\un - \n}{\r - \k} \log \k \\
            & \quad + \frac{1}{2} \binom{\n - 1}{\k - 1}\binom{\un - \n}{\r - \k + 1} \log (\k (\n - \k + 1)) \\
            & \quad + \frac{\n + \k}{2\n} \binom{\n}{\k}\binom{\un - \n}{\r - \k} \binaryEntropy\left(\frac{2\k}{\n + \k}\right) \bigg].
        \end{split}
    \end{equation}
    The claimed leakage follows by plugging~\cref{eq:s2r_entropy_total} into~\cref{eq:s2r_leakage}. For $\r = 0$, only the term with $\k = 1$ in the second line of~\cref{eq:s2r_entropy_total} is nonzero and gives $\entropy(\indexRV \mid \queryRV) = (\log \n)/2$. For $\r = \un$, only the term with $\k = \n$ in the first line is nonzero and gives $\entropy(\indexRV \mid \queryRV) = \log \n$.
\end{proof}

\begin{corollary}
    \label{cor:scenario2fixedrapprox}
    For $\n\r/\un \gg 1$ and $\n(\un - \r)/\un \gg 1$, the leakage in~\cref{thm:scenario2constantr} is approximately
    \begin{align*}
        \leakage
        & \approx \frac{\r}{\un} + \frac{\un - \r}{2\un} \log \frac{\un}{\r} - \frac{\un + \r}{2\un} \log \frac{\un + \r}{\un} \\
        & = \frac{\un - \r}{2\un} \KL{\Bern(1)}{\Bern\left(\frac{\r}{\un}\right)} \\
        & \quad + \frac{\un + \r}{2\un} \KL{\Bern\left(\frac{2\r}{\un + \r}\right)}{\Bern\left(\frac{\r}{\un}\right)}.
    \end{align*}
\end{corollary}
\begin{proof}
    By $\binom{\un - 1}{\r - 1} = (\r/\un)\binom{\un}{\r}$ and $\binom{\un - 1}{\r} = ((\un - \r)/\un)\binom{\un}{\r}$, the weights of the three terms in~\cref{eq:s2r_entropy_total} are $\r/\un$, $(\un - \r)/(2\un)$ and $(\n + \k)/(2\n)$ times hypergeometric probability mass functions of $\k - 1$, $\k - 1$ and $\k$, respectively. For large $\n$, all of them concentrate around $\k = \n\r/\un$. We set $\k = \n\r/\un$, which gives $(\n + \k)/(2\n) = (\un + \r)/(2\un)$ and $2\k/(\n + \k) = 2\r/(\un + \r)$. Using $\k(\n - \k + 1) \approx \k(\n - \k)$,
    \begin{equation*}
        \begin{split}
            \entropy(\indexRV \mid \queryRV)
            & \approx \frac{\r}{\un} \log \frac{\n\r}{\un} + \frac{\un - \r}{2\un} \log \frac{\n^2 \r (\un - \r)}{\un^2} \\
            & \quad + \frac{\un + \r}{2\un} \binaryEntropy\left(\frac{2\r}{\un + \r}\right).
        \end{split}
    \end{equation*}
    The coefficients of $\log \n$ add up to $\r/\un + (\un - \r)/\un = 1$, so $\log \n$ cancels in~\cref{eq:s2r_leakage}. With
    \begin{equation*}
        \frac{\un + \r}{2\un} \binaryEntropy\left(\frac{2\r}{\un + \r}\right) = - \frac{\r}{\un} \log \frac{2\r}{\un + \r} - \frac{\un - \r}{2\un} \log \frac{\un - \r}{\un + \r},
    \end{equation*}
    we obtain
    \begin{align*}
        \leakage
        & \approx \frac{\r}{\un} \log \frac{2\un}{\un + \r} + \frac{\un - \r}{2\un} \log \frac{\un^2}{\r(\un + \r)} \\
        & = \frac{\r}{\un} + \frac{\un - \r}{2\un} \log \frac{\un}{\r} - \frac{\un + \r}{2\un} \log \frac{\un + \r}{\un},
    \end{align*}
    where the last step uses 
    \[
    \frac{\r}{\un} + \frac{\un - \r}{2\un} = \frac{\un + \r}{2\un}.
    \] 
    The KL divergence form follows from 
    \[
    \KL{\Bern(1)}{\Bern(\frac{\r}{\un})} = \log \left(\frac{\un}{\r}\right)
    \]
    and
    \begin{multline*}
        \KL{\Bern\left(\frac{2\r}{\un + \r}\right)}{\Bern\left(\frac{\r}{\un}\right)} \\ 
        = \frac{2\r}{\un + \r} \log \left(\frac{2\r}{\un + \r}\right) - \log \left(\frac{\un + \r}{\un}\right).
    \end{multline*}
\end{proof}
As in~\cref{cor:scenario1fixedrapprox}, a server utilizes the size of its query. With probability $(\un - \r)/(2\un)$, it observes $\size{\queryPool} = \r + 1$ and knows that $\barcodeOfInterest$ is in its query. Otherwise, it observes $\size{\queryPool} = \r$, and $\barcodeOfInterest$ is in its query with probability $2\r/(\un + \r)$ by~\cref{eq:s2r_in_given_r}. In both cases, any other barcode is in its query with probability $\r/\un$.

\begin{theorem}
    Let $\randomPool \subset \universalSet$ be a random pool of sequences and for any sequence $\bfr \in \universalSet$, let $\Pr[\bfr \in \randomPool] = \p$. For a given $\lambda>0$, the failure probability of our protocol naively repeated over $\m \ge 2\lambda / \log(1/ \p)$ servers is at most $2^{-\lambda}$. For any $j \in [2]$, the leakage $\leakage = \mutualInfo{(\queryRV_{j};\indexRV)}$ is 
    \begin{equation*}
        \begin{split}
            \leakage
            & = \frac{1}{2} \log \n - \frac{1}{2} \sum_{\k = 1}^{\n} \binom{\n - 1}{\k - 1} \p^{\k - 1} (1 - \p)^{\n - \k} \log \k \\
            & \quad - \frac{1}{2} \sum_{\k = 1}^{\n} \left[\binom{\n}{\k} \p^{\k} (1 - \p)^{\n - \k} \right. \\
            & \qquad \left. + \binom{\n - 1}{\k - 1} \p^{\k - 1} (1 - \p)^{\n - \k}\right] \binaryEntropy\left(\frac{\k}{\n\p + \k}\right) \\
            & \quad + \frac{1 + \p}{2} \binaryEntropy\left(\frac{\p}{1 + \p}\right).
        \end{split}
    \end{equation*}
    In particular, $\leakage = (\log \n)/2$ for $\p = 0$ and $\leakage = 0$ for $\p = 1$.
    \label{thm:scenario2fixedp}
\end{theorem}
\begin{proof}
    The user will succeed in retrieving the queried record if at least one of the $(\randomPool, \randomPoolStar)$ pairs is different. This is due to the fact that $\randomPool \neq \randomPoolStar$ guarantees $\randomPoolStar \setminus \randomPool = \{\barcodeOfInterest\}$ since $\randomPool = \randomPoolCopy$. So,
    \begin{equation*}
        \Pr[\text{failure}] = \Pr[\barcodeOfInterest \in \randomPool]^{\m / 2} = \p^{\m/2}.
    \end{equation*}
    For a negligible failure probability, we write $2^{-\lambda} \geq \p ^{\m/2}$
    and obtain $\m \ge 2\lambda / \log (1/\p)$. 

    For the leakage, we start with
    \begin{align}\label{eq:s2p_leakage}
        \mutualInfo{(\queryRV; \indexRV)}
        & = \entropy(\indexRV) - \entropy(\indexRV \mid \queryRV) \nonumber \\
        & = \log \n - \entropy(\indexRV \mid \queryRV),
    \end{align}
    where $\entropy(\indexRV) = \log \n$ is due to $\indexRV$ being uniform and $\size{\barcodePool} = \n$. With $\k = \size{\barcodePool \cap \queryPool}$ for a query $\queryPool$, we can break down the second term as
    \begin{equation}\label{eq:s2p_split}
        \entropy(\indexRV \mid \queryRV) = \sum_{\k = 0}^{\n} \Pr[\size{\barcodePool \cap \queryRV} = \k] \cdot \entropy(\indexRV \mid \queryRV, \k).
    \end{equation}
    If $\queryRV = \randomPool$, each of the $\n$ elements of $\barcodePool$ is in $\queryRV$ with probability $\p$. If $\queryRV = \randomPoolStar$, $\barcodeOfInterest \in \queryRV$ and each of the remaining $\n - 1$ elements of $\barcodePool$ is in $\queryRV$ with probability $\p$. So,
    \begin{align}\label{eq:s2p_prob_k}
        & \Pr[\size{\barcodePool \cap \queryRV} = \k] \nonumber \\
        & = \Pr[\size{\barcodePool \cap \queryRV} = \k \mid \queryRV = \randomPool] \cdot \Pr[\queryRV = \randomPool] \nonumber \\
        & \quad + \Pr[\size{\barcodePool \cap \queryRV} = \k \mid \queryRV = \randomPoolStar] \cdot \Pr[\queryRV = \randomPoolStar] \nonumber \\
        & = \frac{1}{2} \binom{\n}{\k} \p^{\k} (1 - \p)^{\n - \k} + \frac{1}{2} \binom{\n - 1}{\k - 1} \p^{\k - 1} (1 - \p)^{\n - \k},
    \end{align}
    where the second term is zero for $\k = 0$. Then, using Bayes',
    \begin{equation}\label{eq:s2p_tube_given_k}
        \Pr[\queryRV = \randomPool \mid \k] = \frac{\n\p}{\n\p + \k}
    \end{equation}
    and
    \[
        \Pr[\queryRV = \randomPoolStar \mid \k] = \frac{\k}{\n\p + \k}.
    \]

    For the entropy term of~\cref{eq:s2p_split}, given $\queryRV = \randomPool$, the query is independent of $\indexRV$, and any of the $\n$ elements of $\barcodePool$ can be the barcode of interest with equal probability. Given $\queryRV = \randomPoolStar$, any of the $\k$ elements of $\barcodePool \cap \queryPool$ can be the barcode of interest with equal probability.
    However, the server does not know which pool it received. Let $\choiceRV \define \indicator{\queryRV = \randomPool}$ denote whether the server receives $\randomPool$. Since $\choiceRV$ is not determined by $\indexRV$ and $\queryRV$, applying the chain rule to $\entropy(\indexRV, \choiceRV \mid \queryRV, \k)$ in two ways gives
    \begin{equation}\label{eq:s2p_chain}
        \entropy(\indexRV \mid \queryRV, \k) = \entropy(\choiceRV \mid \queryRV, \k) + \entropy(\indexRV \mid \queryRV, \choiceRV, \k) - \entropy(\choiceRV \mid \indexRV, \queryRV, \k).
    \end{equation}
    Not knowing $\choiceRV$ introduces the first term, which is by~\cref{eq:s2p_tube_given_k}
    \begin{equation}\label{eq:s2p_entropy_tube}
        \entropy(\choiceRV \mid \queryRV, \k) = \binaryEntropy\left(\frac{\k}{\n\p + \k}\right).
    \end{equation}
    If the server knew $\choiceRV$, its uncertainty would be the second term,
    \begin{equation}\label{eq:s2p_entropy_given_tube}
        \entropy(\indexRV \mid \queryRV, \choiceRV, \k)
        = \frac{\n\p}{\n\p + \k} \log \n + \frac{\k}{\n\p + \k} \log \k.
    \end{equation}
    For the last term, given $\indexRV = \f$, there is no uncertainty about $\choiceRV$ if $\barcodeOfInterest \notin \queryRV$, since the server can then only have received $\randomPool$. The event $\barcodeOfInterest \in \queryRV$ happens in three ways: $\queryRV = \randomPool$ and $\barcodeOfInterest \in \randomPool$ with probability $\p/2$, $\queryRV = \randomPoolStar$ and $\barcodeOfInterest \in \randomPool$ with probability $\p/2$, and $\queryRV = \randomPoolStar$ and $\barcodeOfInterest \notin \randomPool$ with probability $(1 - \p)/2$. Since $\barcodeOfInterest \in \randomPool$ independently of the other elements of $\randomPool$, these probabilities do not depend on the rest of the query, and given $\barcodeOfInterest \in \queryRV$, the server received $\randomPool$ with probability $\p / (1 + \p)$. Moreover, given $\queryRV = \randomPool$, $\barcodeOfInterest \in \queryRV$ with probability $\k / \n$, and given $\queryRV = \randomPoolStar$, $\barcodeOfInterest \in \queryRV$. With~\cref{eq:s2p_tube_given_k},
    \begin{equation*}
        \Pr[\barcodeOfInterest \in \queryRV \mid \k]
        = \frac{\n\p}{\n\p + \k} \cdot \frac{\k}{\n} + \frac{\k}{\n\p + \k} = \frac{\k (1 + \p)}{\n\p + \k},
    \end{equation*}
    and thus
    \begin{equation}\label{eq:s2p_tube_given_fq}
        \entropy(\choiceRV \mid \indexRV, \queryRV, \k) = \frac{\k (1 + \p)}{\n\p + \k} \binaryEntropy\left(\frac{\p}{1 + \p}\right).
    \end{equation}
    Plugging~\cref{eq:s2p_chain} into~\cref{eq:s2p_split}, we write
    \begin{equation}\label{eq:s2p_three}
        \entropy(\indexRV \mid \queryRV) = \entropy_1 + \entropy_2 - \entropy_3,
    \end{equation}
    where $\entropy_1$, $\entropy_2$ and $\entropy_3$ are the three terms on the RHS of~\cref{eq:s2p_chain}, each weighted with $\Pr[\size{\barcodePool \cap \queryRV} = \k]$ and summed over $\k$. For their computation, we use that by~\cref{eq:s2p_tube_given_k,eq:s2p_prob_k},
    \begin{align}\label{eq:s2p_joint}
        \Pr[\size{\barcodePool \cap \queryRV} = \k] \cdot \frac{\n\p}{\n\p + \k} & = \frac{1}{2} \binom{\n}{\k} \p^{\k} (1 - \p)^{\n - \k}, \nonumber \\
        \Pr[\size{\barcodePool \cap \queryRV} = \k] \cdot \frac{\k}{\n\p + \k} & = \frac{1}{2} \binom{\n - 1}{\k - 1} \p^{\k - 1} (1 - \p)^{\n - \k},
    \end{align}
    which are the probabilities that $\size{\barcodePool \cap \queryRV} = \k$ and that the server receives $\randomPool$ and $\randomPoolStar$, respectively. Both sum up to $1/2$ over $\k$.
    With~\cref{eq:s2p_entropy_tube,eq:s2p_prob_k}, and since the term with $\k = 0$ vanishes due to $\binaryEntropy(0) = 0$,
    \begin{align*}
        \entropy_1
        & = \sum_{\k = 0}^{\n} \Pr[\size{\barcodePool \cap \queryRV} = \k] \cdot \binaryEntropy\left(\frac{\k}{\n\p + \k}\right) \\
        & = \frac{1}{2} \sum_{\k = 1}^{\n} \left[\binom{\n}{\k} \p^{\k} (1 - \p)^{\n - \k} \right. \\
        & \qquad \left. + \binom{\n - 1}{\k - 1} \p^{\k - 1} (1 - \p)^{\n - \k}\right] \binaryEntropy\left(\frac{\k}{\n\p + \k}\right).
    \end{align*}
    With~\cref{eq:s2p_entropy_given_tube,eq:s2p_joint}, and since the coefficient of $\log \k$ vanishes for $\k = 0$,
    \begin{align*}
        \entropy_2
        & = \sum_{\k = 0}^{\n} \Pr[\size{\barcodePool \cap \queryRV} = \k] \left[\frac{\n\p}{\n\p + \k} \log \n + \frac{\k}{\n\p + \k} \log \k\right] \\
        & = \frac{1}{2} \log \n + \frac{1}{2} \sum_{\k = 1}^{\n} \binom{\n - 1}{\k - 1} \p^{\k - 1} (1 - \p)^{\n - \k} \log \k.
    \end{align*}
    With~\cref{eq:s2p_tube_given_fq,eq:s2p_joint},
    \begin{align*}
        \entropy_3
        & = \sum_{\k = 0}^{\n} \Pr[\size{\barcodePool \cap \queryRV} = \k] \cdot \frac{\k (1 + \p)}{\n\p + \k} \binaryEntropy\left(\frac{\p}{1 + \p}\right) \\
        & = \frac{1 + \p}{2} \binaryEntropy\left(\frac{\p}{1 + \p}\right).
    \end{align*}
    Bringing it all together,~\cref{eq:s2p_three} becomes
    \begin{equation}\label{eq:s2p_entropy_total}
        \begin{split}
            \entropy(\indexRV \mid \queryRV)
            & = \frac{1}{2} \log \n \\
            & \quad + \frac{1}{2} \sum_{\k = 1}^{\n} \binom{\n - 1}{\k - 1} \p^{\k - 1} (1 - \p)^{\n - \k} \log \k \\
            & \quad + \frac{1}{2} \sum_{\k = 1}^{\n} \left[\binom{\n}{\k} \p^{\k} (1 - \p)^{\n - \k} \right. \\
            & \qquad \left. + \binom{\n - 1}{\k - 1} \p^{\k - 1} (1 - \p)^{\n - \k}\right] \binaryEntropy\left(\frac{\k}{\n\p + \k}\right) \\
            & \quad - \frac{1 + \p}{2} \binaryEntropy\left(\frac{\p}{1 + \p}\right),
        \end{split}
    \end{equation}
    and the claimed leakage follows by plugging~\cref{eq:s2p_entropy_total} into~\cref{eq:s2p_leakage}, with the convention $0^0 = 1$. For $\p = 0$, the server receives $\emptyset$ or $\{\barcodeOfInterest\}$, all terms of the sums vanish, and $\entropy(\indexRV \mid \queryRV) = (\log \n)/2$. For $\p = 1$, only the terms with $\k = \n$ remain, and~\cref{eq:s2p_entropy_total} gives $\entropy(\indexRV \mid \queryRV) = (\log \n)/2 + (\log \n)/2 + 1 - 1 = \log \n$.
\end{proof}

\begin{corollary}
    \label{cor:scenario2approx}
    For $\n\p \gg 1$, the leakage in~\cref{thm:scenario2fixedp} is approximately
    \begin{align*}
        \leakage
        & \approx \frac{1 + \p}{2} \log \frac{1 + \p}{\p} - 1 \\
        & = \KL{\Bern\left(\frac{1 + \p}{2}\right)}{\Bern(\p)}.
    \end{align*}
\end{corollary}
\begin{proof}
    In~\cref{thm:scenario2fixedp}, the weights of the first sum are the probability mass function of $\k - 1 \sim \Bin(\n - 1, \p)$, and the two weights in the second sum are the probability mass functions of $\k \sim \Bin(\n, \p)$ and $\k - 1 \sim \Bin(\n - 1, \p)$. For large $\n$, all of them concentrate around $\k = \n\p$. Setting $\k = \n\p$, the first sum becomes $\log \n\p$ and the second sum becomes $2\binaryEntropy(1/2) = 2$, so
    \begin{align*}
        \leakage
        & \approx \frac{1}{2} \log \n - \frac{1}{2} \log \n\p - 1 + \frac{1 + \p}{2} \binaryEntropy\left(\frac{\p}{1 + \p}\right) \\
        & = \frac{1 + \p}{2} \log \frac{1 + \p}{\p} - 1,
    \end{align*}
    where we used $(1 + \p) \binaryEntropy(\p / (1 + \p)) = \p \log ((1 + \p)/\p) + \log (1 + \p)$. The KL divergence form follows from 
    \begin{align*}
        & \KL{\Bern\left(\frac{1 + \p}{2}\right)}{\Bern(\p)} \\
        & \quad = \frac{1 + \p}{2} \log \frac{1 + \p}{2\p} + \frac{1 - \p}{2} \log \frac{1}{2}. \qedhere
    \end{align*}
\end{proof}
With probability $1/2$ the server receives $\randomPool$, which contains $\barcodeOfInterest$ with probability $\p$, and with probability $1/2$ it receives $\randomPoolStar$, which always contains $\barcodeOfInterest$. So $\barcodeOfInterest$ is in the query with probability $\p/2 + 1/2 = (1 + \p)/2$.

The leakage in~\cref{thm:scenario2constantr} for $\lbc = 20$ is plotted in Fig.~\ref{fig:scenario2}. 
Given $\queryPool$, the server can count $q \define \size{\queryPool}$, and if $q = \r + 1$, deduce that $\barcodeOfInterest \in \queryPool$. 
This happens with probability $(1 - \r/\un)/2$ and reveals that $\barcodeOfInterest$ is one of the $\k$ elements of $\barcodePool \cap \queryPool$.
The leakage diminishes as $\r$ increases.
The leakage in~\cref{thm:scenario2fixedp} for $\lbc = 20$ is also plotted in Fig.~\ref{fig:scenario2}.
Since~\cref{thm:scenario2fixedp} effectively randomizes $\r$, its leakage is smaller than the one in~\cref{thm:scenario2constantr}.
Both leakages converge to $0$ bits as $\r/\un \to 1$ and $\p \to 1$, where retrieval is not possible. 
The markers show the number of servers $\m$ needed to achieve a failure probability of at most $2^{-40}$.

\begin{figure}[ht]
    \centering
    \includegraphics[width=\colfigwidth{0.8}{1}]{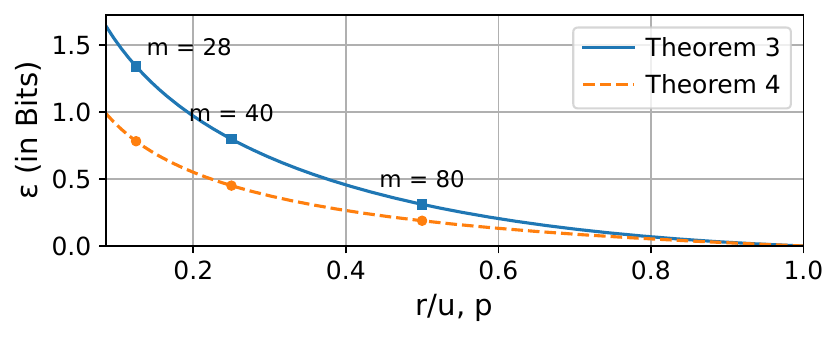}
    \caption{Leakage in Scenario 2 for fixed $\r$ (\cref{thm:scenario2constantr}) and fixed $\p$ (\cref{thm:scenario2fixedp}). Markers show the number of servers $\m$ for failure probability at most $2^{-40}$. $\un \approx 1.1 \cdot 10^{12}, \n \approx 2.7 \cdot 10^{9}$. $\entropy(\indexRV) = 31.36$. The x-axis range is restricted to where the leakage is smaller than 1 bit for~\Cref{thm:scenario2fixedp}.}
    \label{fig:scenario2}
\end{figure}
\subsection{Scenario 3: Uncheckable Pool, Imperfect Copy}
\label{sec:scenario3}
In this scenario, the user is unable to check whether $\barcodeOfInterest \in \randomPool$, and the copy $\randomPoolCopy$ of $\randomPool$ is not perfect. Just as in Scenario~2, the user must repeat the protocol over $\m / 2$ server pairs to bound the failure probability. For each answer pair, the user can simply verify whether they decode to the correct record by assessing $\hbarcode{\f} = \barcode{\f}$. If that is not the case, the user can discard the output and decode the answer pair of the next server pair.

\section{Discussion and Future Work}
In this work, we established a foundational framework for PIR in DNA-based data storage. A key property of our protocols is utilizing random synthesis, which is significantly cheaper than synthesizing specific sequences. Several interesting challenges remain, e.g., more realistic experimental noise is considered in our ongoing work as our current formulations largely assume noiseless biochemical operations. Specifically, we are investigating the effects of imperfect PCR, where a primer designed for barcode $\barcode{i}$ inadvertently amplifies unintended records $\rec{j}$ (with $i\ne j$) due to barcode/primer similarity. The impact of imperfect pool copying, briefly introduced in Scenario 3, is also a focus of current extensions. Moreover, configurations with more than two servers are of interest. The relationship between the synthesis and sampling parameters and the resulting effective $\p$ is left as future work.

This research direction aims to improve the functionality, flexibility, and privacy aspects of DNA-based data storage. 



\bibliographystyle{IEEEtran}
\bibliography{aux/bibliography.bib}

\end{document}